\documentclass{amsart}
\newcommand{\Z}{\mathbb Z}
\newcommand{\R}{\mathbb R}
 \usepackage{epsfig} \usepackage[usenames,dvipsnames]{pstricks}
\renewcommand{\P}{\mathbb P}
\newtheorem{thm}{Theorem}[section]
\newtheorem{prop}[thm]{Proposition}

\newtheorem{lemma}[thm]{Lemma}

\newtheorem{cor}[thm]{Corollary}
\newtheorem{remark}[thm]{Remark}
\usepackage{amsmath}
\usepackage{amsxtra}

\usepackage{amscd}
\usepackage{amsthm}
\usepackage{amsfonts}
\usepackage{amssymb}
\usepackage{eucal}

\numberwithin{figure}{section}

\numberwithin{equation}{section}

\newcommand{\bx}{\mathbf x}

\renewcommand{\k}{\kappa}

\renewcommand{\sl}{{\mathfrak{sl}}}

\renewcommand{\mod}{\rm ~mod~}

\begin{document}
\title{$T$-systems and the pentagram map}
\author{Rinat Kedem}\footnote{rinat@illinois.edu, fax: +1-217-333-9576}
\address{Department of Mathematics, University of Illinois, Urbana, IL 61821 USA}
\address{RK: rinat@illinois.edu}
\address{PV: vichitk1@illinois.edu}
\author{Panupong Vichitkunakorn}
\date{\today}

\begin{abstract}
  These notes summarize two different connections between two discrete
  integrable systems, the $A_d$ $T$-system and its infinite-rank
  analog, the octahedron relation, and the pentagram map and its
  various generalizations.
\end{abstract}
\maketitle
\section{Introduction}
The purpose of this article is to summarize two connections between (generalized) pentagram maps and the octahedron relation with special boundary conditions, also known as the $T$-system. The first connection is essentially the one found by Glick \cite{Glick} and by Gekhtman et. al. \cite{GSTV} in the two-dimensional case, and we clarify here the exact relation between the variables of octahedron relation and the various coordinates used in describing the pentagram maps. Essentially the pentagram map is the $Y$-system corresponding to the usual octahedron relation with initial conditions wrapped on a torus. We show how to unfold the $Y$-system so that the result is a quasi-periodic $T$-system on the same torus. The discrete dynamics of the pentagram map is inherited from the usual octahedron map. 

The second connection to $T$-systems is related to the Zamolodchikov (quasi-) periodicity phenomenon for the $A_d$ $T$-systems. It relates the solutions of the $T$-system exhibiting the Zamolodchikov periodicity phenomenon
to the lift of the projective coordinates of the polyhedron in projective $d$-space. It gives a new interpretation to the coefficients appearing in the linear recursion relation satisfied by the lifted coordinates in terms of generalized $q$-characters of $U_q(\widehat{\mathfrak sl}_{d+1})$.

The paper is organized as follows. In Section 1, we introduce the octahedron relation (the $T$-system with no boundary conditions) and the solutions using networks found in \cite{DFKT}. In Section 2 we introduce the generalized pentagram maps \cite{GSTV} and show how this map is related to the octahedron relation with quasi-periodic initial data, or its related $Y$-system with initial data wrapped on a torus. We also introduce the linear recursion relation satisfied by the lifted coordinates of the $n$-gon in projective $d$-space. In Section 3, we return to this linear recursion relation, and relate its coefficients as conserved quantities of the $A$-type $T$-system with wall type boundary conditions. Periodicity of the coefficients is a direct result of Zamolodchikov periodicity.

\vskip.1in
\noindent{\bf Acknowledgements} This research is supported in part by NSF grant DMS-1100929. We thank Philippe Di Francesco, Gloria Mar{\'{\i}} Beffa, Valentin Ovsienko and Richard Schwartz for helpful discussions. RK thanks the organizers of the conference FDIS13 where these results were presented.
\section{The octahedron relation and $T$-systems}
\subsection{A bilinear recursion relation}
Let $T$ be a function defined on the vertices of the lattice $\Z^3$. The octahedron relation is a recursion relation between the values of the function $T$ on a sublattice of $\Z^3$:
\begin{equation}\label{Tsys}
T_{i,j,k+1} T_{i,j,k-1} = T_{i,j+1,k} T_{i,j-1,k} + T_{i+1,j,k}T_{i-1,j,k}.
\end{equation}
For each choice of integers $(i,j,k)$, this is a relation between the values of $T$ on the vertices of an octahedron in $\Z^3$ centered at $(i,j,k)$. 
It relates even and odd vertices in $\Z^3$ separately, that is, those with $i+j+k=0\mod 2$ or $1\mod 2$. Without loss of generality, we may concentrate on one of the two sublattices in our discussion.

We consider the octahedron relation to be a discrete evolution. Given a set of initial data on a valid initial data surface ${\mathbf k} = (i,j,k(i,j))$, the fuction $T$ is determined everywhere else in $\Z^3$. By a valid initial data surface we mean one where $|k(i,j)-k(i\pm1,j)|=1$ and $|k(i,j)-k(i,j\pm1)|=1$ for all $i,j\in \Z$. Here we may assume that $i+j+k(i,j)$ is odd, for example. An example of a valid initial data surface is $(i,j,i+j+1\mod 2)$.

\subsection{Cluster algebra structure}
The octahedron relation appears in numerous contexts in combinatorics,  statistical mechanics and representation theory.  In the context in which it is used here, it is interesting to note its relation to cluster algebras in particular.

In fact, it is natural to
interpret the octahedron relation as a {\em mutation} in a coefficient-free cluster algebra \cite{FZ} of geometric type and infinite rank
\cite{DFK09}. The cluster algebra itself is most easily specified in terms of an initial quiver, which encodes the exchange matrix associated with a particular type of cluster seed data. We call this seed data {\em initial data} for the octahedron relation. 

For example, on the even sublattice all values $T$ are determined by the initial data $\{T_{i,j,i+j+1\mod 2}: i,j\in\Z\}$. Application of the octahedron relation give all other values of $T$. If we write $x_{i,j}= T_{i,j,i+j+1\mod 2}$ then for each $j,k$, the mutation
$$\mu_{i,j}(x_{i,j})= x'_{i,j}= \frac{x_{i,j+1}x_{i,j-1}+x_{i+1,j}x_{i-1,j}}{x_{i,j}}$$ is encoded by the exchange matrix 
\begin{equation}\label{bmatrix}
B_{i,j}^{i',j'} = (-1)^{i+j}(\delta_{i,i'}\delta_{j,j'\pm1} - \delta_{j,j'}\delta_{i,i'\pm1}).
\end{equation}
We illustrate the quiver associated with the initial exchange matrix $B$:
\begin{figure}
\input{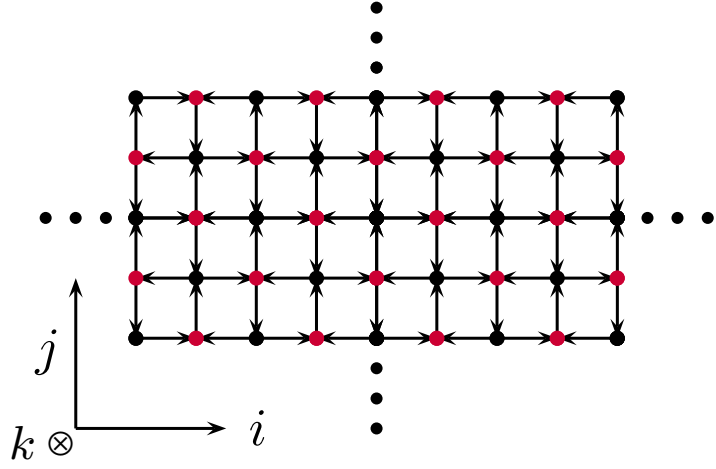}
\caption{The quiver associated with the exchange matrix $B$ of the octahedron relation for the initial data $T_{i,j,i+j+1\mod 2}$. The red dots correspond to the $k=0$ plane and the black ones to the $k=1$ plane. The quiver is infinite in the $i$ and $j$ directions.}
\label{octaquiver}
\end{figure}

All this just means that the octahedron relation is written as
$$x'_{i,j} x_{i,j} = \prod_{i',j'} x_{i',j'}^{[B_{i,j}^{i',j'}]_+} + \prod_{i',j'} x_{i',j'}^{[-B_{i,j}^{i',j'}]_+},
$$
where $[n]_+$ is the positive part of $n$. We interpret the new variable $x'_{i,j}$ to be $T_{i,j,2}$ if $i+j$ is odd, or $T_{i,j,-1}$ if $i+j$ is even.

The theorem \cite{DFK09} is as follows:
\begin{thm}
The octahedron relations are mutations in a cluster algebra which contains the initial cluster  consisting of the cluster variables $\{T_{i,j,i+j+1\mod 2}: i,j\in\Z\}$ and the exchange matrix $B$ of Equation \eqref{bmatrix}.
\end{thm}
The subset of mutations among all cluster algebra mutations which correspond to one of the octahedron relations correspond to a mutation at a vertex of the quiver which has two incoming and two outgoing arrows. Mutations of the quiver itself create other types of vertices but repeated applications of the octahedron relations locally restore the original quiver.

\subsection{Solutions of the octahedron relations via network matrices}

Solving the octahedron relation means providing an explicit expression for any variable $T_{i',j',k}$ (assume without loss of generality that $i'+j'+k\equiv 1\mod 2$) in terms of 
the variables of the initial data $\bx_0=\{T_{i,j,i+j+1\mod 2}: i,j\in\Z \}$. In fact, since the octahedron relation is a mutation in a cluster algebra, the solutions are all cluster variables and hence are Laurent polynomials (instead of just rational functions) \cite{FZ} of a finite subset of the initial data variables in $\bx_0$. 

Moreover, by explicitly solving the system, we can show that its coefficients are non-negative integers \cite{DFKT}, in line with the general positivity conjecture of cluster algebras.

We give an expression for the explicit solution in the particularly simple case we consider here. This version of the solution is found in \cite{DF} and we do refer the reader to that resource for the details of this solution. 

Assume $k>1$ (for $k<0$, a symmetric argument holds). Then we can visualize the point $i',j',k$ as a point in a plane above the initial data plane $\{i,j,i+j+1\mod 2, i,j\in \Z\}$. The initial data surface is in bijection with a weighted network, whose edge weights are monomials in the initial data, as follows:
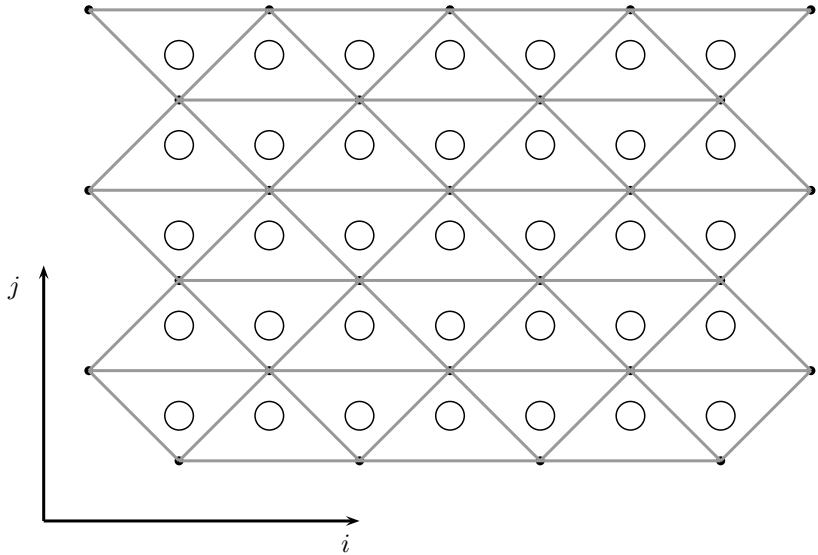
\begin{figure}
% Generated with LaTeXDraw 2.0.8
% Mon Jan 06 17:52:55 CET 2014
% \usepackage[usenames,dvipsnames]{pstricks}
% \usepackage{epsfig}
% \usepackage{pst-grad} % For gradients
% \usepackage{pst-plot} % For axes
\scalebox{1} % Change this value to rescale the drawing.
{
\begin{pspicture}(0,-3.683711)(10.861015,3.6637108)
\definecolor{color97}{rgb}{0.6,0.6,0.6}
\psdots[dotsize=0.12](1.1810156,3.583711)
\psdots[dotsize=0.12](3.5810156,3.583711)
\psdots[dotsize=0.12](5.9810157,3.583711)
\psdots[dotsize=0.12](4.7810154,2.3837109)
\psdots[dotsize=0.12](8.381016,3.583711)
\psdots[dotsize=0.12](10.781015,3.583711)
\psdots[dotsize=0.12](2.3810155,2.3837109)
\psdots[dotsize=0.12](7.1810155,2.3837109)
\psdots[dotsize=0.12](9.581016,2.3837109)
\psdots[dotsize=0.12](1.1810156,1.1837109)
\psdots[dotsize=0.12](3.5810156,1.1837109)
\psdots[dotsize=0.12](5.9810157,1.1837109)
\psdots[dotsize=0.12](4.7810154,-0.016289063)
\psdots[dotsize=0.12](8.381016,1.1837109)
\psdots[dotsize=0.12](10.781015,1.1837109)
\psdots[dotsize=0.12](2.3810155,-0.016289063)
\psdots[dotsize=0.12](7.1810155,-0.016289063)
\psdots[dotsize=0.12](9.581016,-0.016289063)
\psdots[dotsize=0.12](1.1810156,-1.216289)
\psdots[dotsize=0.12](3.5810156,-1.216289)
\psdots[dotsize=0.12](5.9810157,-1.216289)
\psdots[dotsize=0.12](4.7810154,-2.416289)
\psdots[dotsize=0.12](8.381016,-1.216289)
\psdots[dotsize=0.12](10.781015,-1.216289)
\psdots[dotsize=0.12](2.3810155,-2.416289)
\psdots[dotsize=0.12](7.1810155,-2.416289)
\psdots[dotsize=0.12](9.581016,-2.416289)
\psline[linewidth=0.04cm,linecolor=color97](1.1810156,3.583711)(10.781015,3.583711)
\psline[linewidth=0.04cm,linecolor=color97](2.3810155,2.3837109)(9.581016,2.3837109)
\psline[linewidth=0.04cm,linecolor=color97](1.1810156,1.1837109)(10.781015,1.1837109)
\psline[linewidth=0.04cm,linecolor=color97](2.3810155,-0.016289063)(9.581016,-0.016289063)
\psline[linewidth=0.04cm,linecolor=color97](1.1810156,-1.216289)(10.781015,-1.216289)
\psline[linewidth=0.04cm,linecolor=color97](2.3810155,-2.416289)(9.581016,-2.416289)
\psline[linewidth=0.04cm,linecolor=color97](3.5810156,3.583711)(1.1810156,1.1837109)
\psline[linewidth=0.04cm,linecolor=color97](5.9810157,3.583711)(1.1810156,-1.216289)
\psline[linewidth=0.04cm,linecolor=color97](8.381016,3.583711)(2.3810155,-2.416289)
\psline[linewidth=0.04cm,linecolor=color97](10.781015,3.583711)(4.7810154,-2.416289)
\psline[linewidth=0.04cm,linecolor=color97](10.781015,1.1837109)(7.1810155,-2.416289)
\psline[linewidth=0.04cm,linecolor=color97](10.781015,-1.216289)(9.581016,-2.416289)
\psline[linewidth=0.04cm,linecolor=color97](3.5810156,3.583711)(9.581016,-2.416289)
\psline[linewidth=0.04cm,linecolor=color97](5.9810157,3.583711)(9.581016,-0.016289063)
\psline[linewidth=0.04cm,linecolor=color97](8.381016,3.583711)(10.781015,1.1837109)
\psline[linewidth=0.04cm,linecolor=color97](1.1810156,3.583711)(7.1810155,-2.416289)
\psline[linewidth=0.04cm,linecolor=color97](1.1810156,1.1837109)(4.7810154,-2.416289)
\psline[linewidth=0.04cm,linecolor=color97](1.1810156,-1.216289)(2.3810155,-2.416289)
\psline[linewidth=0.04cm,linecolor=color97](9.581016,-0.016289063)(10.781015,-1.216289)
\psdots[dotsize=0.4,fillstyle=solid,dotstyle=o](2.3810155,-1.8162891)
\psdots[dotsize=0.4,fillstyle=solid,dotstyle=o](3.5810156,-1.8162891)
\psdots[dotsize=0.4,fillstyle=solid,dotstyle=o](4.7810154,-1.8162891)
\psdots[dotsize=0.4,fillstyle=solid,dotstyle=o](5.9810157,-1.8162891)
\psdots[dotsize=0.4,fillstyle=solid,dotstyle=o](7.1810155,-1.8162891)
\psdots[dotsize=0.4,fillstyle=solid,dotstyle=o](8.381016,-1.8162891)
\psdots[dotsize=0.4,fillstyle=solid,dotstyle=o](9.581016,-1.8162891)
\psdots[dotsize=0.4,fillstyle=solid,dotstyle=o](9.581016,-0.6162891)
\psdots[dotsize=0.4,fillstyle=solid,dotstyle=o](8.381016,-0.6162891)
\psdots[dotsize=0.4,fillstyle=solid,dotstyle=o](7.1810155,-0.6162891)
\psdots[dotsize=0.4,fillstyle=solid,dotstyle=o](5.9810157,-0.6162891)
\psdots[dotsize=0.4,fillstyle=solid,dotstyle=o](4.7810154,-0.6162891)
\psdots[dotsize=0.4,fillstyle=solid,dotstyle=o](3.5810156,-0.6162891)
\psdots[dotsize=0.4,fillstyle=solid,dotstyle=o](2.3810155,-0.6162891)
\psdots[dotsize=0.4,fillstyle=solid,dotstyle=o](2.3810155,0.5837109)
\psdots[dotsize=0.4,fillstyle=solid,dotstyle=o](3.5810156,0.5837109)
\psdots[dotsize=0.4,fillstyle=solid,dotstyle=o](4.7810154,0.5837109)
\psdots[dotsize=0.4,fillstyle=solid,dotstyle=o](5.9810157,0.5837109)
\psdots[dotsize=0.4,fillstyle=solid,dotstyle=o](7.1810155,0.5837109)
\psdots[dotsize=0.4,fillstyle=solid,dotstyle=o](8.381016,0.5837109)
\psdots[dotsize=0.4,fillstyle=solid,dotstyle=o](9.581016,0.5837109)
\psdots[dotsize=0.4,fillstyle=solid,dotstyle=o](2.3810155,1.783711)
\psdots[dotsize=0.4,fillstyle=solid,dotstyle=o](3.5810156,1.783711)
\psdots[dotsize=0.4,fillstyle=solid,dotstyle=o](4.7810154,1.783711)
\psdots[dotsize=0.4,fillstyle=solid,dotstyle=o](5.9810157,1.783711)
\psdots[dotsize=0.4,fillstyle=solid,dotstyle=o](7.1810155,1.783711)
\psdots[dotsize=0.4,fillstyle=solid,dotstyle=o](8.381016,1.783711)
\psdots[dotsize=0.4,fillstyle=solid,dotstyle=o](9.581016,1.783711)
\psdots[dotsize=0.4,fillstyle=solid,dotstyle=o](2.3810155,2.983711)
\psdots[dotsize=0.4,fillstyle=solid,dotstyle=o](3.5810156,2.983711)
\psdots[dotsize=0.4,fillstyle=solid,dotstyle=o](4.7810154,2.983711)
\psdots[dotsize=0.4,fillstyle=solid,dotstyle=o](5.9810157,2.983711)
\psdots[dotsize=0.4,fillstyle=solid,dotstyle=o](7.1810155,2.983711)
\psdots[dotsize=0.4,fillstyle=solid,dotstyle=o](8.381016,2.983711)
\psdots[dotsize=0.4,fillstyle=solid,dotstyle=o](9.581016,2.983711)
\psline[linewidth=0.04cm,arrowsize=0.05291667cm 2.0,arrowlength=1.4,arrowinset=0.4]{->}(0.58101565,-3.216289)(4.7810154,-3.216289)
\psline[linewidth=0.04cm,arrowsize=0.05291667cm 2.0,arrowlength=1.4,arrowinset=0.4]{->}(0.58101565,-3.216289)(0.58101565,0.18371093)
\usefont{T1}{ptm}{m}{n}
\rput(4.5924706,-3.5112891){$i$}
\usefont{T1}{ptm}{m}{n}
\rput(0.1924707,-0.11128906){$j$}
\end{pspicture} 
}
\caption{The network corresponding to the initial data surface. The initial data $T_{i,j,i+j+1\mod 2}$ is situated at the white circles on the lattice faces. The network is the triangular lattice, where all edges are assumed to be oriented from left to right. Its weights are given in terms of the face variables.}
\label{fig:network}
\end{figure}
The edge weights are given in terms of the face variables as follows:
\begin{figure}
% Generated with LaTeXDraw 2.0.8
% Mon Jan 06 18:14:36 CET 2014
% \usepackage[usenames,dvipsnames]{pstricks}
% \usepackage{epsfig}
% \usepackage{pst-grad} % For gradients
% \usepackage{pst-plot} % For axes
\scalebox{1} % Change this value to rescale the drawing.
{
\begin{pspicture}(0,-1.82)(4.82,1.82)
\psline[linewidth=0.04cm](2.4,0.6)(1.2,-0.6)
\psline[linewidth=0.04cm](1.2,-0.6)(3.6,-0.6)
\psline[linewidth=0.04cm](2.4,0.6)(3.6,-0.6)
\psline[linewidth=0.04cm,linestyle=dashed,dash=0.16cm 0.16cm](2.4,0.6)(3.6,1.8)
\psline[linewidth=0.04cm,linestyle=dashed,dash=0.16cm 0.16cm](2.4,0.6)(4.8,0.6)
\psline[linewidth=0.04cm,linestyle=dashed,dash=0.16cm 0.16cm](1.2,1.8)(2.4,0.6)
\psline[linewidth=0.04cm,linestyle=dashed,dash=0.16cm 0.16cm](1.2,1.8)(3.6,1.8)
\psline[linewidth=0.04cm,linestyle=dashed,dash=0.16cm 0.16cm](4.8,0.6)(2.4,-1.8)
\psline[linewidth=0.04cm,linestyle=dashed,dash=0.16cm 0.16cm](2.4,-1.8)(0.0,0.6)
\psline[linewidth=0.04cm,linestyle=dashed,dash=0.16cm 0.16cm](0.0,0.6)(2.4,0.6)
\psline[linewidth=0.04cm,linestyle=dashed,dash=0.16cm 0.16cm](1.2,-0.6)(0.0,-1.8)
\psline[linewidth=0.04cm,linestyle=dashed,dash=0.16cm 0.16cm](0.0,-1.8)(2.4,-1.8)
\usefont{T1}{ptm}{m}{n}
\rput(1.273125,-1.295){a}
\usefont{T1}{ptm}{m}{n}
\rput(1.2241992,0.145){b}
\usefont{T1}{ptm}{m}{n}
\rput(2.3878515,1.365){c}
\usefont{T1}{ptm}{m}{n}
\rput(3.564463,0.165){d}
\usefont{T1}{ptm}{m}{n}
\rput(2.3881152,-1.015){e}
\usefont{T1}{ptm}{m}{n}
\rput(2.3912792,-0.095){f}
\psdots[dotsize=0.6,fillstyle=solid,dotstyle=square](2.94,0.06)
\usefont{T1}{ptm}{m}{n}
\rput(2.9214551,0.065){$\frac{c}{d}$}
\psdots[dotsize=0.6,fillstyle=solid,dotstyle=square](1.86,0.08)
\usefont{T1}{ptm}{m}{n}
\rput(1.8314551,0.105){$\frac{a}{f}$}
\psdots[dotsize=0.6,fillstyle=solid,dotstyle=square](2.38,-0.58)
\usefont{T1}{ptm}{m}{n}
\rput(2.4314551,-0.555){$\frac{ab}{ef}$}
\end{pspicture} 
}
\caption{The edge weights of the network in terms of the face variables. The variables $a,b,c$ etc. are the initial data $T_{i,j,i+j\mod 2}$ and the weights are positive Laurent monomials in these variables.}
\label{fig:weights}
\end{figure}
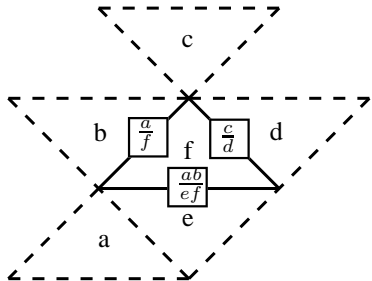
\begin{figure}
% Generated with LaTeXDraw 2.0.8
% Tue Jan 07 15:03:38 CET 2014
% \usepackage[usenames,dvipsnames]{pstricks}
% \usepackage{epsfig}
% \usepackage{pst-grad} % For gradients
% \usepackage{pst-plot} % For axes
\scalebox{1} % Change this value to rescale the drawing.
{
\begin{pspicture}(0,-5.33)(10.359043,5.31)
\definecolor{color712}{rgb}{0.6,0.6,0.6}
\psdots[dotsize=0.12](5.1427736,4.45)
\psdots[dotsize=0.12](3.9427733,3.25)
\psdots[dotsize=0.12](6.3427734,3.25)
\psdots[dotsize=0.12](2.7427735,-2.75)
\psdots[dotsize=0.12](2.7427735,2.05)
\psdots[dotsize=0.12](5.1427736,2.05)
\psdots[dotsize=0.12](3.9427733,0.85)
\psdots[dotsize=0.12](7.5427732,2.05)
\psdots[dotsize=0.12](3.9427733,-3.95)
\psdots[dotsize=0.12](1.5427735,0.85)
\psdots[dotsize=0.12](6.3427734,0.85)
\psdots[dotsize=0.12](8.742773,0.85)
\psdots[dotsize=0.12](0.34277344,-0.35)
\psdots[dotsize=0.12](2.7427735,-0.35)
\psdots[dotsize=0.12](5.1427736,-0.35)
\psdots[dotsize=0.12](3.9427733,-1.55)
\psdots[dotsize=0.12](7.5427732,-0.35)
\psdots[dotsize=0.12](9.942774,-0.35)
\psdots[dotsize=0.12](1.5427735,-1.55)
\psdots[dotsize=0.12](6.3427734,-1.55)
\psdots[dotsize=0.12](8.742773,-1.55)
\psline[linewidth=0.04cm,linecolor=color712](3.9427733,3.25)(6.3427734,3.25)
\psline[linewidth=0.04cm,linecolor=color712](2.7427735,2.05)(7.5427732,2.05)
\psline[linewidth=0.04cm,linecolor=color712](1.5427735,0.85)(8.742773,0.85)
\psline[linewidth=0.04cm,linecolor=color712](0.34277344,-0.35)(9.942774,-0.35)
\psline[linewidth=0.04cm,linecolor=color712](1.5427735,-1.55)(8.742773,-1.55)
\psline[linewidth=0.04cm,linecolor=color712](5.1427736,4.45)(0.34277344,-0.35)
\psline[linewidth=0.04cm,linecolor=color712](7.5427732,2.05)(2.7427735,-2.75)
\psline[linewidth=0.04cm,linecolor=color712](8.742773,0.85)(3.9427733,-3.95)
\psline[linewidth=0.04cm,linecolor=color712](9.942774,-0.35)(5.1427736,-5.15)
\psline[linewidth=0.04cm,linecolor=color712](3.9427733,3.25)(8.742773,-1.55)
\psline[linewidth=0.04cm,linecolor=color712](5.1427736,4.45)(8.742773,0.85)
\psline[linewidth=0.04cm,linecolor=color712](2.7427735,2.05)(7.5427732,-2.75)
\psline[linewidth=0.04cm,linecolor=color712](1.5427735,0.85)(6.3427734,-3.95)
\psline[linewidth=0.04cm,linecolor=color712](0.34277344,-0.35)(5.1427736,-5.15)
\psline[linewidth=0.04cm,linecolor=color712](8.742773,0.85)(9.942774,-0.35)
\psdots[dotsize=0.4,fillstyle=solid,dotstyle=o](1.5427735,-0.95)
\psdots[dotsize=0.4,fillstyle=solid,dotstyle=o](2.7427735,-0.95)
\psdots[dotsize=0.4,fillstyle=solid,dotstyle=o](3.9427733,-0.95)
\psdots[dotsize=0.4,fillstyle=solid,dotstyle=o](5.1427736,-0.95)
\psdots[dotsize=0.4,fillstyle=solid,dotstyle=o](6.3427734,-0.95)
\psdots[dotsize=0.4,fillstyle=solid,dotstyle=o](7.5427732,-0.95)
\psdots[dotsize=0.4](8.742773,-0.95)
\psdots[dotsize=0.4,fillstyle=solid,dotstyle=o](8.742773,0.25)
\psdots[dotsize=0.4,fillstyle=solid,dotstyle=o](7.5427732,0.25)
\psdots[dotsize=0.4,fillstyle=solid,dotstyle=o](6.3427734,0.25)
\psdots[dotsize=0.4,dotstyle=otimes](5.1427736,0.25)
\psdots[dotsize=0.4,fillstyle=solid,dotstyle=o](3.9427733,0.25)
\psdots[dotsize=0.4,fillstyle=solid,dotstyle=o](2.7427735,0.25)
\psdots[dotsize=0.4,fillstyle=solid,dotstyle=o](1.5427735,0.25)
\psdots[dotsize=0.4,fillstyle=solid,dotstyle=o](2.7427735,1.45)
\psdots[dotsize=0.4,fillstyle=solid,dotstyle=o](3.9427733,1.45)
\psdots[dotsize=0.4,fillstyle=solid,dotstyle=o](5.1427736,1.45)
\psdots[dotsize=0.4,fillstyle=solid,dotstyle=o](6.3427734,1.45)
\psdots[dotsize=0.4,fillstyle=solid,dotstyle=o](7.5427732,1.45)
\psdots[dotsize=0.4,fillstyle=solid,dotstyle=o](3.9427733,2.65)
\psdots[dotsize=0.4,fillstyle=solid,dotstyle=o](5.1427736,2.65)
\psdots[dotsize=0.4,fillstyle=solid,dotstyle=o](6.3427734,2.65)
\psdots[dotsize=0.4,fillstyle=solid,dotstyle=o](5.1427736,3.85)
\psline[linewidth=0.04cm,linecolor=color712](6.3427734,3.25)(1.5427735,-1.55)
\psdots[dotsize=0.12](5.1427736,-5.15)
\psdots[dotsize=0.12](6.3427734,-3.95)
\psdots[dotsize=0.12](7.5427732,-2.75)
\psdots[dotsize=0.4,fillstyle=solid,dotstyle=o](3.9427733,-2.15)
\psdots[dotsize=0.4,fillstyle=solid,dotstyle=o](5.1427736,-2.15)
\psdots[dotsize=0.4,fillstyle=solid,dotstyle=o](6.3427734,-2.15)
\psline[linewidth=0.04cm,linecolor=color712](2.7427735,-2.75)(7.5427732,-2.75)
\psline[linewidth=0.04cm,linecolor=color712](3.9427733,-3.95)(6.3427734,-3.95)
\psdots[dotsize=0.4,fillstyle=solid,dotstyle=o](2.7427735,-2.15)
\psdots[dotsize=0.4](7.5427732,-2.15)
\psdots[dotsize=0.4,fillstyle=solid,dotstyle=o](3.9427733,-3.35)
\psdots[dotsize=0.4,fillstyle=solid,dotstyle=o](5.1427736,-3.35)
\psdots[dotsize=0.4](6.3427734,-3.35)
\psdots[dotsize=0.4](5.1427736,-4.55)
\usefont{T1}{ptm}{m}{n}
\rput(0.05229492,-0.385){5}
\usefont{T1}{ptm}{m}{n}
\rput(1.3037207,-1.685){4}
\usefont{T1}{ptm}{m}{n}
\rput(2.5305371,-2.865){3}
\usefont{T1}{ptm}{m}{n}
\rput(3.7214355,-4.045){2}
\usefont{T1}{ptm}{m}{n}
\rput(4.789707,-5.165){1}
\usefont{T1}{ptm}{m}{n}
\rput(5.489707,-5.185){1}
\usefont{T1}{ptm}{m}{n}
\rput(6.6014357,-4.025){2}
\usefont{T1}{ptm}{m}{n}
\rput(7.770537,-2.885){3}
\usefont{T1}{ptm}{m}{n}
\rput(8.96372,-1.705){4}
\usefont{T1}{ptm}{m}{n}
\rput(10.232295,-0.445){5}
\psdots[dotsize=0.4](9.942774,0.29)
\psdots[dotsize=0.4,fillstyle=solid,dotstyle=o](8.802773,1.55)
\psdots[dotsize=0.4,fillstyle=solid,dotstyle=o](7.4427733,2.71)
\psdots[dotsize=0.4,fillstyle=solid,dotstyle=o](6.3427734,3.91)
\psdots[dotsize=0.4,fillstyle=solid,dotstyle=o](3.9227734,3.87)
\psdots[dotsize=0.4,fillstyle=solid,dotstyle=o](2.7427735,2.67)
\psdots[dotsize=0.4,fillstyle=solid,dotstyle=o](1.5627735,1.51)
\psdots[dotsize=0.4,fillstyle=solid,dotstyle=o](0.34277344,0.25)
\psdots[dotsize=0.4,fillstyle=solid,dotstyle=o](5.1227736,5.09)
\end{pspicture} 
}
\caption{The network which contributes to the expression $T_{i,j,5}$ a distance $k=5$ above the initial data surface. The marked circle at the center is the point $i,j$ on the surface in the $i$-$j$-plane.}
\label{fig:diamond}
\end{figure}
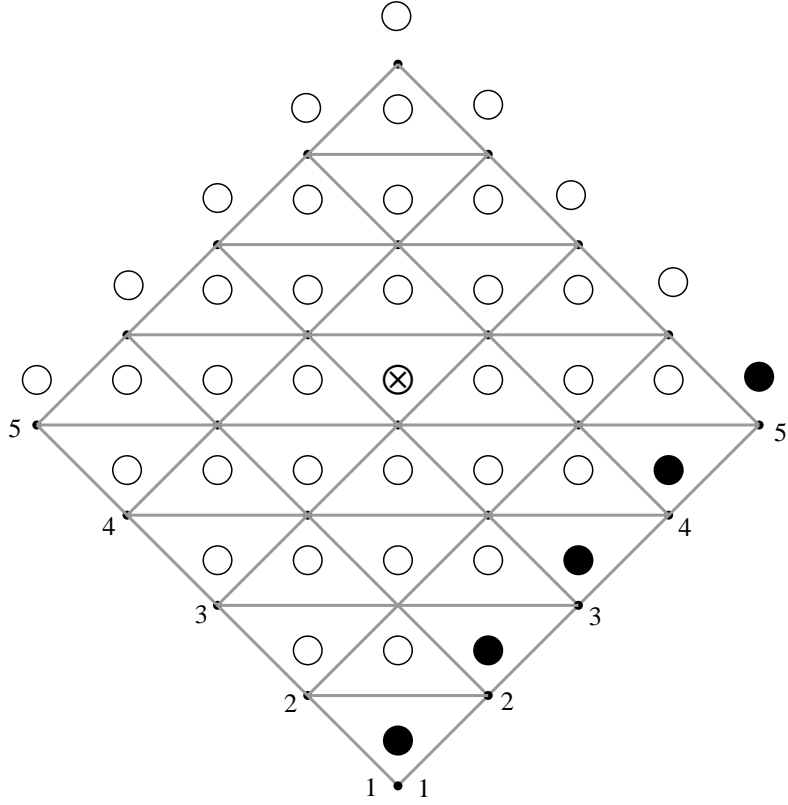
The point $i,j,k$ with $i+j+k$ odd is situated above the point $i,j,i+j+1\mod 2$ on the initial data plane. Consider a diamond-shaped subset of the initial data plane centered at this point, with sides of length $k$, and its corresponding weighted network. 

The network matrix of size $k\times k$, $N_k(i,j)$, is the matrix whose  $a,b$ entry is the partition function of paths on the weighted matrix from point $a$ on the left of the network in Figure \ref{fig:diamond} to point $b$ on the right. The determinant of this matrix is the partition function of $k$ non-intersecting paths on the network from the points $(1,...,k)$ on the left to the points $(1,...,k)$ on the right \cite{GV}. The partition function is just the sum of products of edge weights, hence a positive Laurent polynomial in the initial data.

Finally the variable $T_{i,j,k}$ is the determinant  of $N_k(i,j)$, multiplied by the $k$ variables denoted by black dots along the southeast edge of the diamond in Figure \ref{fig:diamond}.
\subsection{Associated $Y$-system}
There is a closely related discrete evolution to the $T$-system or the octahedron relation called the $Y$-system. It was originally encountered in the context of thermedynamic Bethe ansatz in the 80's \cite{Z} as a relation among the fugacities of quasiparticles in conformal field theories. It can be obtained directly from the $T$-system as follows \cite{KNS}.

Define the variables 
$$
Y_{i,j,k} = \frac{T_{i+1,j,k}T_{i-1,j,k}}{T_{i,j+1,k}T_{i,j-1,k}}.
$$
Dividing both sides of the octahedron relation \eqref{Tsys} by $T_{i+1,j,k}T_{i-1,j,k}$ and multiplying the resulting equations for two separate values of $(i,j,k)$, we get a so-called $Y$-system:
\begin{equation}\label{Ysys}
Y_{i,j,k+1}Y_{i,j,k-1} = \frac{(1+Y_{i+1,j,k})(1+Y_{i-1,j,k})}{(1+Y^{-1}_{i,j+1,k})(1+Y^{-1}_{i,j-1,k})}.
\end{equation}

\begin{remark} The original $Y$-system associated with the Lie algebra $A_d$ has a restricted set of values for $i$ with the appropriate boundary conditions, that is $1\leq i \leq d$. The system \eqref{Ysys} is again a recursion relation in $\Z^3$. 
\end{remark}

The $Y$-system \eqref{Ysys} is identified with the mutation relation for {\em coefficients} ($Y$-variables in the language of \cite{FZ} or $x$-variables in the laguage of \cite{FG}.) The exchange matrix $B$ or the quiver associated with it are the same as for the $T$-system. The birational expression for $Y$ in terms of the $T$-variables means that the map between the two sets of variables is not bijective but projective. Therefore, given a $Y$-system, there are many choices of ``unfolding" it back into a $T$-system. We will give explicit formulas for sufficiently general such unfoldings in the next section.
\section{The higher pentagram map: coordinates, cluster algebra and Y-system}
\subsection{The (higher) pentagram map}
The pentagram map is a discrete evolution equation acting on points in $\R\P^2$. In its original version it was introduced by Richard Schwartz in series
of papers \cite{Schwartz,Schwartz2,Schwartz3}. See also \cite{OST} for a concise review. This map acts on a (twisted) polygon with $n$ vertices to give another twisted $n$-gon, whose vertices are the intersections of the short diagonals of the original polygon (connecting next-nearest neighbor vertices modulo $n$). See Fig.~\ref{Fig:pentagrammap} for an example.

A twisted $n-$gon is a sequence of vertices $(v_{i})_{i\in\mathbb{Z}}$ in
$\mathbb{RP}^{2}$ together with a monodromy $M$, a projective automorphism,
such that $v_{i+n}=M\circ v_{i}$ for all $i\in\mathbb{Z}.$ Denote by $\mathcal{P}_{n}$
the space of projective equivalence classes of twisted $n-$gons, and let $T:\mathcal{P}_{n}\rightarrow\mathcal{P}_{n}$ denote
the pentagram map. The $i$-th vertex of the image, $T(v_{i}),$ is
then defined to be the intersection of the two diagonals $\overline{v_{i-1}v_{i+1}}$
and $\overline{v_{i}v_{i+2}},$ i.e., 
\[
T(v_{i})=\overline{v_{i-1}v_{i+1}}\cap\overline{v_{i}v_{i+2}}.
\]
as in Fig.~\ref{Fig:localpenta}.

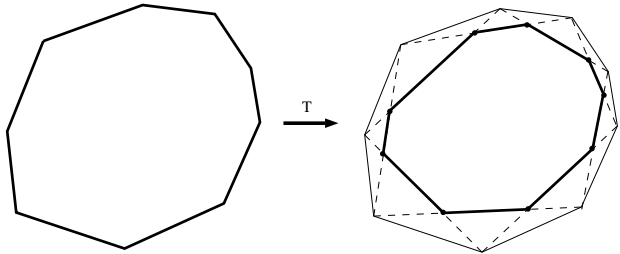
\begin{figure}
% Generated with LaTeXDraw 2.0.8
% Mon Dec 16 14:29:46 CST 2013
% \usepackage[usenames,dvipsnames]{pstricks}
% \usepackage{epsfig}
% \usepackage{pst-grad} % For gradients
% \usepackage{pst-plot} % For axes
\scalebox{0.6} % Change this value to rescale the drawing.
{
\begin{pspicture}(0,-2.75)(13.53,2.77)
\psline[linewidth=0.02](8.72,1.86)(10.92,2.66)(12.52,2.46)(13.32,1.26)(13.52,0.06)(12.72,-1.74)(10.52,-2.74)(8.12,-1.94)(7.92,-0.14)(8.72,1.86)
\psline[linewidth=0.02cm,linestyle=dashed,dash=0.16cm 0.16cm](10.92,2.66)(13.32,1.26)
\psline[linewidth=0.02cm,linestyle=dashed,dash=0.16cm 0.16cm](13.52,0.06)(10.52,-2.74)
\psline[linewidth=0.02cm,linestyle=dashed,dash=0.16cm 0.16cm](12.72,-1.74)(8.12,-1.94)
\psline[linewidth=0.02cm,linestyle=dashed,dash=0.16cm 0.16cm](10.52,-2.74)(7.92,-0.14)
\psline[linewidth=0.02cm,linestyle=dashed,dash=0.16cm 0.16cm](8.12,-1.94)(8.72,1.86)
\psline[linewidth=0.02cm,linestyle=dashed,dash=0.16cm 0.16cm](7.92,-0.14)(10.92,2.66)
\psline[linewidth=0.02cm,linestyle=dashed,dash=0.16cm 0.16cm](8.72,1.86)(12.52,2.46)
\psdots[dotsize=0.12](10.36,2.12)
\psdots[dotsize=0.12](11.52,2.3)
\psdots[dotsize=0.12](12.88,1.52)
\psdots[dotsize=0.12](12.96,-0.44)
\psdots[dotsize=0.12](11.54,-1.78)
\psdots[dotsize=0.12](9.66,-1.86)
\psdots[dotsize=0.12](8.32,-0.56)
\psdots[dotsize=0.12](8.48,0.38)
\psline[linewidth=0.06](0.8,1.94)(3.0,2.74)(4.6,2.54)(5.4,1.34)(5.6,0.14)(4.8,-1.66)(2.6,-2.66)(0.2,-1.86)(0.0,-0.06)(0.8,1.94)
\psline[linewidth=0.08cm,arrowsize=0.05291667cm 2.0,arrowlength=1.4,arrowinset=0.0]{->}(6.12,0.12)(7.34,0.12)
\usefont{T1}{ptm}{m}{n}
\rput(6.64,0.49){T}
\psline[linewidth=0.02cm,linestyle=dashed,dash=0.16cm 0.16cm](13.31,1.27)(12.73,-1.79)
\psline[linewidth=0.02cm,linestyle=dashed,dash=0.16cm 0.16cm](12.51,2.43)(13.49,0.09)
\psdots[dotsize=0.12](13.22,0.74)
\pspolygon[linewidth=0.06](8.33,-0.55)(8.47,0.39)(10.35,2.13)(11.53,2.31)(12.89,1.51)(13.21,0.77)(12.97,-0.45)(11.53,-1.79)(9.67,-1.87)
\end{pspicture} 
}
\caption{The pentagram map on a closed polygon}
\label{Fig:pentagrammap}
\end{figure}

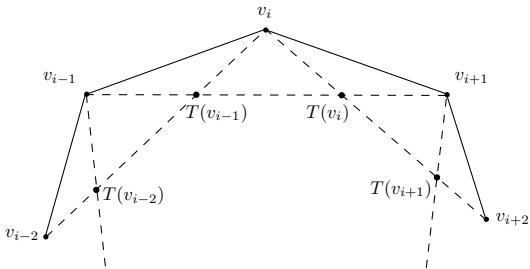
\begin{figure}
% Generated with LaTeXDraw 2.0.8
% Fri Dec 20 12:53:54 CST 2013
% \usepackage[usenames,dvipsnames]{pstricks}
% \usepackage{epsfig}
% \usepackage{pst-grad} % For gradients
% \usepackage{pst-plot} % For axes
\scalebox{0.7} % Change this value to rescale the drawing.
{
\begin{pspicture}(0,-2.5142188)(10.822812,2.5442188)
\psline[linewidth=0.02](1.1295311,-1.9153187)(1.8913449,0.78758126)(5.300462,2.0257812)(8.747668,0.80268127)(9.509483,-1.6133188)
\psline[linewidth=0.02,linestyle=dashed,dash=0.16cm 0.16cm](1.1485765,-1.9153187)(5.319507,2.0257812)(9.509483,-1.5982188)
\psline[linewidth=0.02cm,fillcolor=black,linestyle=dashed,dash=0.16cm 0.16cm](8.747668,0.77248126)(1.8913449,0.78758126)
\psdots[dotsize=0.1](5.3195314,2.0157812)
\psdots[dotsize=0.1](8.759531,0.78578126)
\psdots[dotsize=0.1](1.9195313,0.79578125)
\psdots[dotsize=0.1](9.509531,-1.5842187)
\psdots[dotsize=0.1](1.1395313,-1.9142188)
\psdots[dotsize=0.12](8.569531,-0.7842187)
\psdots[dotsize=0.12](3.9995313,0.78578126)
\psline[linewidth=0.02cm,fillcolor=black,linestyle=dashed,dash=0.16cm 0.16cm](1.9103903,0.78758126)(2.272252,-2.4740188)
\psline[linewidth=0.02cm,fillcolor=black,linestyle=dashed,dash=0.16cm 0.16cm](8.747668,0.71208125)(8.366762,-2.5042188)
\psdots[dotsize=0.12](2.0995312,-1.0242188)
\psdots[dotsize=0.12](6.7595315,0.7757813)
\usefont{T1}{ptm}{m}{n}
\rput(5.302344,2.3557813){$v_i$}
\usefont{T1}{ptm}{m}{n}
\rput(9.232344,1.0957812){$v_{i+1}$}
\usefont{T1}{ptm}{m}{n}
\rput(10.012343,-1.5842187){$v_{i+2}$}
\usefont{T1}{ptm}{m}{n}
\rput(0.6923438,-1.8842187){$v_{i-2}$}
\usefont{T1}{ptm}{m}{n}
\rput(1.4123436,1.0957812){$v_{i-1}$}
\usefont{T1}{ptm}{m}{n}
\rput(6.5023437,0.43578127){$T(v_i)$}
\usefont{T1}{ptm}{m}{n}
\rput(7.8923435,-1.0042187){$T(v_{i+1})$}
\usefont{T1}{ptm}{m}{n}
\rput(4.3923435,0.43578127){$T(v_{i-1})$}
\usefont{T1}{ptm}{m}{n}
\rput(2.8123438,-1.1242187){$T(v_{i-2})$}
\end{pspicture} 
}
\caption{The image of $v_i$ under the pentagram map}
\label{Fig:localpenta}
\end{figure}

More recently, Ovsienko, Schwartz and Tabachnikov proved
integrability of the pentagram map for the space of twisted polygons
\cite{OST} with the help of a pentagram-invariant Poisson structure,
and for the space of closed polygons \cite{OST2}.

There also is a sequence of generalizations of the pentagram map in $\R\P^2$ called ``higher pentagram maps" \cite{GSTV}. This generalization was inspired by the cluster algebra structure of the Schwartz pentagram map, first described by Glick \cite{Glick}. The higher pentagram maps correspond to using intersections of  longer diagonals of the $n$-gon to map $\mathcal P_n$ to itself. 
For a given integer $\k\geq 3,$ the generalized pentagram map $T_\k$ constructs a new polygon using $(\k-1)^{th}-$diagonals (connecting vertex $i$ to vertex $i+\k-1$) instead of the shortest diagonals. Using the Poisson structure compatible with the cluster algebra structure the authors were also able to show the integrability of the higher pentagram maps \cite{GSTV}.

%For this reason, we may consider this map acting on polygons embedded in higher dimensional projective space, $\mathbb{RP}^{k-1}.$   
%
%\begin{figure}
%\input{fig_localpenta_GSTV}
%\caption{The higher pentagram map.}
%\label{Fig: highpenta}
%\end{figure}

Let us introduce two integers $r=\lfloor\frac{\k-2}2\rfloor$ and $r'=\lceil \frac{\k-2}2 \rceil$,
such that 
\begin{eqnarray*}
r+r'+2=\k, \qquad r'= \left\{ \begin{array}{ll}
r & \hbox{$\k$ even}
\\r+1 & \hbox{$\k$ odd.}
\end{array}\right.
\end{eqnarray*}
The higher pentagram map can be expressed in terms the projective invariants $p,q$ 
of the twisted polygon $A$, defined as follows:
\begin{eqnarray*}
p_i(A) &=& -\chi(v_{i-r'},\overline{v_{i-r'-1}v_{i+r}}\cap L,L\cap\overline{v_{i-r'+1}v_{i+r+2}},v_{i+r+1})^{-1}, \nonumber \\
q_i(A) &=& -\chi(\overline{v_{i-\k+1}v_{i-\k+2}}\cap\overline{v_{i}v_{i+1}},v_{i},v_{i+1},\overline{v_{i}v_{i+1}}\cap\overline{v_{i+\k-1}v_{i+\k}})\, .
\end{eqnarray*}
Here, 
$$\chi(a,b,c,d)=\frac{(a-b)(c-d)}{(a-c)(b-d)}.
$$
and $L=\overline{v_{i-r'}v_{i+r+1}}.$
Note that $p_i$ and $q_i$ are periodic with period $n$.
The higher pentagram transformation $A\to T_\k(A)$ reads\footnote{Eq. \ref{highpentaT} is in fact the inverse of the transformation defined in \cite{GSTV}, up to the interchange of the $p$ and $q$ invariants. The latter reads $(p_i,q_i)\mapsto (p_i^*,q_i^*)$, with:
\begin{equation}
q_i^*=p_{i+r-r'}^{-1}\qquad p_i^*=q_i \frac{(1+p_{i-r'-1})(1+p_{i+r+1})}{(1+p_{i-r'}^{-1})(1+p_{i+r}^{-1})}
\end{equation}}:
\begin{equation}\label{highpentaT}
q_i(T_\k(A))=p_{i+r'-r}(A)^{-1}\qquad p_i(T_\k(A))=q_i(A) \frac{(1+p_{i-r}(A))(1+p_{i+r'}(A))}{(1+p_{i-r-1}(A)^{-1})(1+p_{i+r'+1}(A)^{-1})}
\end{equation}

This map was shown by \cite{OST,GSTV,Sol} to be integrable. In particular, there are two useful conserved quantities:
\begin{equation}\label{consT} O_n(A)=\prod_{i=1}^n p_i(A) \qquad {\rm and}\qquad  E_n(A)=\prod_{i=1}^n q_i(A) 
\end{equation}
namely such that $O_n(T_\k(A))=O_n(A)$ and $E_n(T_\k(A))=E_n(A)$.

\subsection{The Schwartz pentagram map:  cluster algebra structure}

Note that the transformation \eqref{highpentaT} reduces to the ordinary \cite{Schwartz} pentagram map $T\equiv T_3$ of Fig. \ref{Fig:localpenta} when $\k=3$, i.e. $r=0$ and $r'=1$.

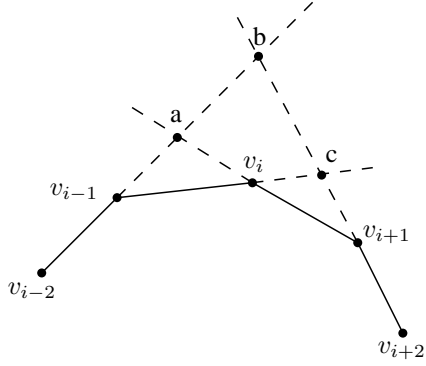
\begin{figure}
% Generated with LaTeXDraw 2.0.8
% Fri Dec 20 12:19:33 CST 2013
% \usepackage[usenames,dvipsnames]{pstricks}
% \usepackage{epsfig}
% \usepackage{pst-grad} % For gradients
% \usepackage{pst-plot} % For axes
\scalebox{1} % Change this value to rescale the drawing.
{
\begin{pspicture}(0,-2.4489062)(6.4228125,2.4189062)
\psdots[dotsize=0.12](0.8195312,-1.1910938)
\psdots[dotsize=0.12](1.8195312,-0.19109374)
\psdots[dotsize=0.12](3.6195314,0.00890625)
\psdots[dotsize=0.12](5.0195312,-0.79109377)
\psdots[dotsize=0.12](5.619531,-1.9910938)
\psline[linewidth=0.02cm](0.8195312,-1.1910938)(1.8195312,-0.19109374)
\psline[linewidth=0.02cm](1.8195312,-0.19109374)(3.6195314,0.00890625)
\psline[linewidth=0.02cm](3.6195314,0.00890625)(5.0195312,-0.79109377)
\psline[linewidth=0.02cm](5.0195312,-0.79109377)(5.619531,-1.9910938)
\psline[linewidth=0.02cm,linestyle=dashed,dash=0.16cm 0.16cm](1.8195312,-0.19109374)(4.4195313,2.4089062)
\psline[linewidth=0.02cm,linestyle=dashed,dash=0.16cm 0.16cm](5.0195312,-0.79109377)(3.4195313,2.2089062)
\psline[linewidth=0.02cm,linestyle=dashed,dash=0.16cm 0.16cm](3.6195314,0.00890625)(5.219531,0.20890625)
\psline[linewidth=0.02cm,linestyle=dashed,dash=0.16cm 0.16cm](3.6195314,0.00890625)(2.0195312,1.0089062)
\psdots[dotsize=0.12](2.6195314,0.60890627)
\psdots[dotsize=0.12](4.539531,0.10890625)
\psdots[dotsize=0.12](3.6995313,1.6889062)
\usefont{T1}{ptm}{m}{n}
\rput(0.6923438,-1.4410938){$v_{i-2}$}
\usefont{T1}{ptm}{m}{n}
\rput(1.2523438,-0.12109375){$v_{i-1}$}
\usefont{T1}{ptm}{m}{n}
\rput(3.6423438,0.25890625){$v_{i}$}
\usefont{T1}{ptm}{m}{n}
\rput(5.412344,-0.68109375){$v_{i+1}$}
\usefont{T1}{ptm}{m}{n}
\rput(5.612344,-2.2210937){$v_{i+2}$}
\usefont{T1}{ptm}{m}{n}
\rput(2.6057813,0.89890623){a}
\usefont{T1}{ptm}{m}{n}
\rput(3.7179687,1.9589063){b}
\usefont{T1}{ptm}{m}{n}
\rput(4.6754684,0.35890624){c}
\end{pspicture} 
}
\caption{Corner invariants: $X_i=\chi(v_{i-2},v_{i-1},a,b),~Y_i=\chi(b,c,v_{i+1},v_{i+2})$}
\label{Fig:corner}
\end{figure}

In his original construction, Schwartz introduced for any 
twisted $n$-gon $A=(v_{i})_{i\in\mathbb{Z}}$ a set of coordinates
$(X_{i},Y_{i})_{i=1}^{n}$ called \emph{corner invariants}, and defined as:
\begin{eqnarray*}
X_{i}(A) & =&\chi(v_{i-2},v_{i-1},\overline{v_{i-2}v_{i-1}}\cap\overline{v_{i}v_{i+1}},\overline{v_{i-2}v_{i-1}}\cap\overline{v_{i+1}v_{i+2}}),\\
Y_{i}(A) & =&\chi(\overline{v_{i-2}v_{i-1}}\cap\overline{v_{i+1}v_{i+2}},\overline{v_{i-1}v_{i}}\cap\overline{v_{i+1}v_{i+2}},v_{i+1},v_{i+2}).
\end{eqnarray*}
These are illustrated in Figure \ref{Fig:corner}. The corner invariants are related to the $p,q$ invariants via:
\begin{equation}
p_i(A) = -\left(X_i(A) Y_i(A)\right)^{-1},\quad q_i(A) = -Y_i(A) X_{i+1}(A)\, .
\end{equation}

Now we can formally relate the pentagram map evolution of $y$-coordinates and the cluster algebra mutation. For a twisted $n-$gon $A$ parameterized by $y$-parameters $(p_i,q_i)_{i=1}^n,$ we define a labeled $Y$-seed corresponding to $A$ to be $(\textbf{y},B)$ where 
\begin{equation}\label{eq: pentasmutn}
\textbf{y}=\left(p_1(A),\dots,p_n(A),q_1(A),\dots,q_n(A)\right)\text,\quad B=\begin{pmatrix}\textbf{0} & C \\ -C^\top & \textbf{0}\end{pmatrix}
\end{equation}
where $C=(c_{ij})$ is an ${n\times n}$ matrix defined by $c_{ij}=\delta_{i,j-1}-\delta_{i,j}-\delta_{i,j+1}+\delta_{i,j+2}$ where the indices are read modulo $n.$ Equivalently, the quiver corresponding to the exchange matrix $B$ is a bipartite graph with $2n$ vertices labeled by $p_{1},\dots,p_{n},q_{1},\dots,q_{n}.$ There are four arrows adjacent to each $q_{i}$: two outgoing arrows from $q_{i}$ to $p_{i}$ and $p_{i+1},$ two incoming arrows from $p_{i-1}$ and $p_{i+2}$ to $q_{i}.$ See an example in Fig.~\ref{Fig: GlickQuiver}. 

The pentagram map is then a composition of a sequence of mutations
on all the $p_{i}-$vertices followed by a relabeling $\{p_{i}\mapsto q_{i+1},q_{i}\mapsto p_{i}\}$\cite{Glick}. This maps the coefficients($y-$parameters) $(p_i(A),q_i(A))$ to $\left(p_i(T(A)),q_i(T(A))\right),$ and the quiver $B$ to ifself.  See an example in Fig.~\ref{Fig: MutationGlick}. 

\begin{figure}
% Generated with LaTeXDraw 2.0.8
% Fri Dec 20 12:52:02 CST 2013
% \usepackage[usenames,dvipsnames]{pstricks}
% \usepackage{epsfig}
% \usepackage{pst-grad} % For gradients
% \usepackage{pst-plot} % For axes
\scalebox{0.8} % Change this value to rescale the drawing.
{
\begin{pspicture}(0,-4.55)(9.482813,4.55)
\usefont{T1}{ptm}{m}{n}
\rput(5.612344,2.24){$q_1$}
\rput{-292.57877}(5.1565895,-5.754008){\pscircle[linewidth=0.04,dimen=outer](6.890449,0.98743486){0.2995644}}
\usefont{T1}{ptm}{m}{n}
\rput(6.8923435,1.0){$q_2$}
\rput{-292.57877}(0.21905442,-4.859106){\pscircle[linewidth=0.04,dimen=outer](3.7510266,-2.26539){0.2995644}}
\usefont{T1}{ptm}{m}{n}
\rput(3.752344,-2.24){$q_5$}
\rput{-292.57877}(0.56061625,-2.8179455){\pscircle[linewidth=0.04,dimen=outer](2.3921258,-0.9888371){0.2995644}}
\usefont{T1}{ptm}{m}{n}
\rput(2.4123437,-0.98){$q_6$}
\rput{-292.57877}(2.310045,-1.6380448){\pscircle[linewidth=0.04,dimen=outer](2.382602,0.91216594){0.2995644}}
\usefont{T1}{ptm}{m}{n}
\rput(2.3923438,0.92){$q_7$}
\rput{-292.57877}(4.348336,-2.0571082){\pscircle[linewidth=0.04,dimen=outer](3.715801,2.2301652){0.2995644}}
\usefont{T1}{ptm}{m}{n}
\rput(3.732344,2.24){$q_8$}
\pscircle[linewidth=0.04,dimen=outer](4.659531,4.25){0.3}
\usefont{T1}{ptm}{m}{n}
\rput(4.6723437,4.26){$p_1$}
\pscircle[linewidth=0.04,dimen=outer](1.6195312,3.01){0.3}
\usefont{T1}{ptm}{m}{n}
\rput(1.6323436,3.02){$p_8$}
\pscircle[linewidth=0.04,dimen=outer](8.999532,0.03){0.3}
\usefont{T1}{ptm}{m}{n}
\rput(8.992344,0.04){$p_3$}
\pscircle[linewidth=0.04,dimen=outer](7.679531,-3.05){0.3}
\usefont{T1}{ptm}{m}{n}
\rput(7.6723437,-3.04){$p_4$}
\pscircle[linewidth=0.04,dimen=outer](4.679531,-4.25){0.3}
\usefont{T1}{ptm}{m}{n}
\rput(4.6923437,-4.24){$p_5$}
\pscircle[linewidth=0.04,dimen=outer](1.6195312,-3.03){0.3}
\usefont{T1}{ptm}{m}{n}
\rput(1.6323436,-3.0){$p_6$}
\pscircle[linewidth=0.04,dimen=outer](0.39953125,-0.01){0.3}
\usefont{T1}{ptm}{m}{n}
\rput(0.41234374,0.0){$p_7$}
\pscircle[linewidth=0.04,dimen=outer](7.679531,3.03){0.3}
\usefont{T1}{ptm}{m}{n}
\rput(7.6923437,3.04){$p_2$}
\psline[linewidth=0.04cm,arrowsize=0.05291667cm 2.0,arrowlength=1.4,arrowinset=0.0]{->}(3.88,2.49)(4.639531,3.97)
\psline[linewidth=0.04cm,arrowsize=0.05291667cm 2.0,arrowlength=1.4,arrowinset=0.0]{->}(5.48,2.45)(4.7195315,3.97)
\psline[linewidth=0.04cm,arrowsize=0.05291667cm 2.0,arrowlength=1.4,arrowinset=0.0]{->}(5.88,2.31)(7.4795313,2.85)
\psline[linewidth=0.04cm,arrowsize=0.05291667cm 2.0,arrowlength=1.4,arrowinset=0.0]{->}(6.98,1.27)(7.5195312,2.81)
\psline[linewidth=0.04cm,arrowsize=0.05291667cm 2.0,arrowlength=1.4,arrowinset=0.0]{->}(7.14,0.85)(8.719531,0.09)
\psline[linewidth=0.04cm,arrowsize=0.05291667cm 2.0,arrowlength=1.4,arrowinset=0.0]{->}(7.14,-0.79)(8.7395315,0.01)
\psline[linewidth=0.04cm,arrowsize=0.05291667cm 2.0,arrowlength=1.4,arrowinset=0.0]{->}(6.98,-1.21)(7.5195312,-2.83)
\psline[linewidth=0.04cm,arrowsize=0.05291667cm 2.0,arrowlength=1.4,arrowinset=0.0]{->}(5.82,-2.37)(7.4595313,-2.91)
\psline[linewidth=0.04cm,arrowsize=0.05291667cm 2.0,arrowlength=1.4,arrowinset=0.0]{->}(5.44,-2.55)(4.7195315,-3.97)
\psline[linewidth=0.04cm,arrowsize=0.05291667cm 2.0,arrowlength=1.4,arrowinset=0.0]{->}(3.88,-2.51)(4.639531,-3.97)
\psline[linewidth=0.04cm,arrowsize=0.05291667cm 2.0,arrowlength=1.4,arrowinset=0.0]{->}(3.46,-2.33)(1.8795314,-2.87)
\psline[linewidth=0.04cm,arrowsize=0.05291667cm 2.0,arrowlength=1.4,arrowinset=0.0]{->}(2.3,-1.25)(1.7995312,-2.83)
\psline[linewidth=0.04cm,arrowsize=0.05291667cm 2.0,arrowlength=1.4,arrowinset=0.0]{->}(2.14,-0.85)(0.65953124,-0.07)
\psline[linewidth=0.04cm,arrowsize=0.05291667cm 2.0,arrowlength=1.4,arrowinset=0.0]{->}(2.14,0.75)(0.6795312,0.01)
\psline[linewidth=0.04cm,arrowsize=0.05291667cm 2.0,arrowlength=1.4,arrowinset=0.0]{->}(2.3,1.19)(1.7995312,2.79)
\psline[linewidth=0.04cm,arrowsize=0.05291667cm 2.0,arrowlength=1.4,arrowinset=0.0]{->}(3.46,2.31)(1.8595313,2.83)
\psline[linewidth=0.04cm,arrowsize=0.05291667cm 2.0,arrowlength=1.4,arrowinset=0.0]{->}(4.539531,4.01)(2.5,1.17)
\psline[linewidth=0.04cm,arrowsize=0.05291667cm 2.0,arrowlength=1.4,arrowinset=0.0]{->}(4.8195314,4.01)(6.74,1.19)
\psline[linewidth=0.04cm,arrowsize=0.05291667cm 2.0,arrowlength=1.4,arrowinset=0.0]{->}(7.3995314,2.93)(3.98,2.31)
\psline[linewidth=0.04cm,arrowsize=0.05291667cm 2.0,arrowlength=1.4,arrowinset=0.0]{->}(7.599531,2.77)(6.94,-0.67)
\psline[linewidth=0.04cm,arrowsize=0.05291667cm 2.0,arrowlength=1.4,arrowinset=0.0]{->}(8.7795315,0.17)(5.84,2.11)
\psline[linewidth=0.04cm,arrowsize=0.05291667cm 2.0,arrowlength=1.4,arrowinset=0.0]{->}(8.7395315,-0.09)(5.78,-2.19)
\psline[linewidth=0.04cm,arrowsize=0.05291667cm 2.0,arrowlength=1.4,arrowinset=0.0]{->}(7.579531,-2.77)(6.96,0.71)
\psline[linewidth=0.04cm,arrowsize=0.05291667cm 2.0,arrowlength=1.4,arrowinset=0.0]{->}(4.7995315,-3.99)(6.78,-1.19)
\psline[linewidth=0.04cm,arrowsize=0.05291667cm 2.0,arrowlength=1.4,arrowinset=0.0]{->}(4.539531,-3.99)(2.58,-1.23)
\psline[linewidth=0.04cm,arrowsize=0.05291667cm 2.0,arrowlength=1.4,arrowinset=0.0]{->}(1.8795314,-2.95)(5.28,-2.35)
\psline[linewidth=0.04cm,arrowsize=0.05291667cm 2.0,arrowlength=1.4,arrowinset=0.0]{->}(1.6995312,-2.77)(2.36,0.65)
\psline[linewidth=0.04cm,arrowsize=0.05291667cm 2.0,arrowlength=1.4,arrowinset=0.0]{->}(0.65953124,-0.15)(3.46,-2.19)
\psline[linewidth=0.04cm,arrowsize=0.05291667cm 2.0,arrowlength=1.4,arrowinset=0.0]{->}(0.6795312,0.13)(3.46,2.09)
\psline[linewidth=0.04cm,arrowsize=0.05291667cm 2.0,arrowlength=1.4,arrowinset=0.0]{->}(1.7195313,2.75)(2.3,-0.75)
\psline[linewidth=0.04cm,arrowsize=0.05291667cm 2.0,arrowlength=1.4,arrowinset=0.0]{->}(1.8595313,2.93)(5.32,2.29)
\rput{-292.57877}(3.3739865,-6.9574337){\pscircle[linewidth=0.04,dimen=outer](6.901016,-0.95019215){0.2995644}}
\usefont{T1}{ptm}{m}{n}
\rput(6.9123435,-0.94){$q_3$}
\rput{-292.57877}(5.5063314,-3.7953463){\pscircle[linewidth=0.04,dimen=outer](5.597465,2.2288682){0.2995644}}
\rput{-292.57877}(1.3143553,-6.518886){\pscircle[linewidth=0.04,dimen=outer](5.5425453,-2.2744422){0.2995644}}
\usefont{T1}{ptm}{m}{n}
\rput(5.5523434,-2.26){$q_4$}
\psline[linewidth=0.04cm,arrowsize=0.05291667cm 2.0,arrowlength=1.4,arrowinset=0.0]{->}(7.3995314,-2.97)(4.0,-2.35)
\end{pspicture} 
}
\caption{Glick's quiver for $n=8$}
\label{Fig: GlickQuiver}
\end{figure}
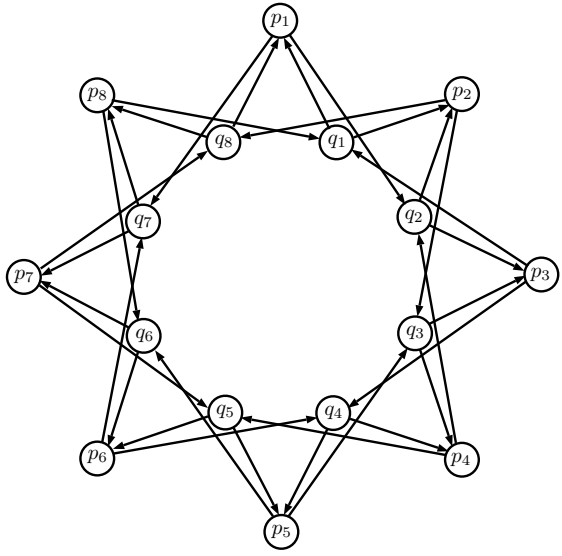

\begin{figure}
\input{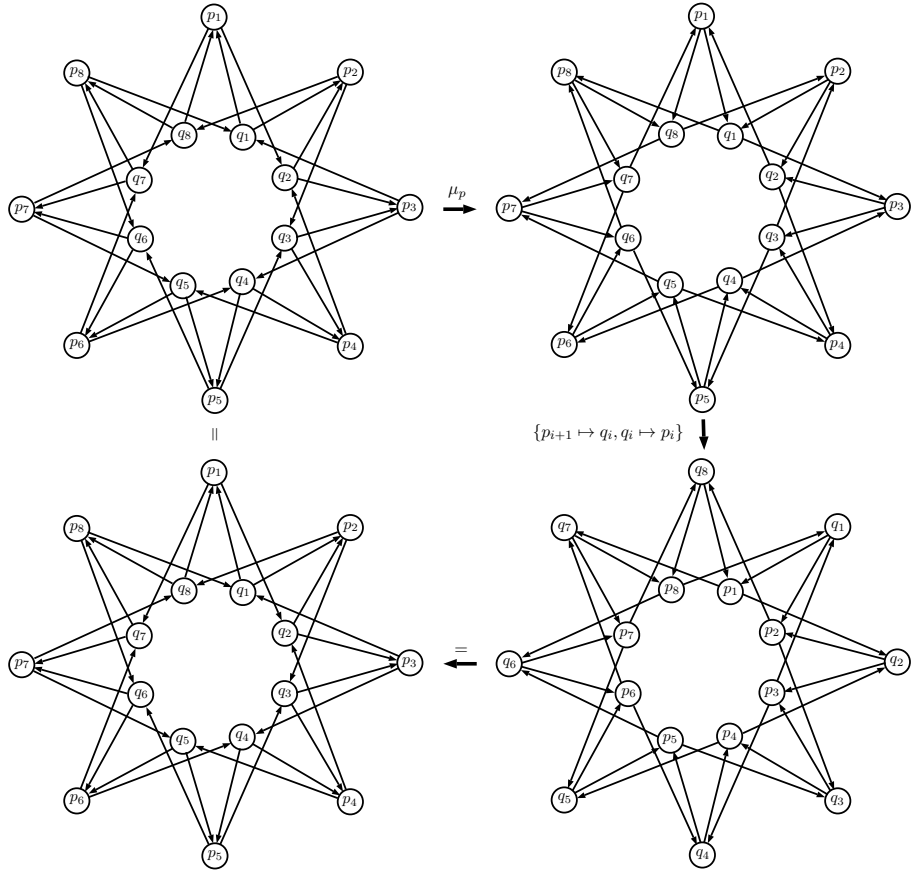}
\caption{The Glick's quiver ($n=8$) after a sequence of mutations on all $p-vertices$ and a relabelling $\{p_{i+1}\mapsto q_i, q_i\mapsto p_i\}$.}
\label{Fig: MutationGlick}
\end{figure}

\subsection{Higher pentagram maps: cluster algebra structure}

Similarly to the case $\k=3$, the evolution \eqref{highpentaT} can be identified with a cluster algebra mutation. 
For a twisted $n-$gon $A$ having $(\mathbf{p,q})-$parameters $(p_i,q_i)_{i=1}^n,$ we define a labeled 
$Y-$seed corresponding to $A$ to be $(\textbf{y},B)$ as in \eqref{eq: pentasmutn} except that the matrix 
$C=(c_{ij})$ is defined to be $c_{ij}=\delta_{i,j-r-1}-\delta_{i,j-r}-\delta_{i,j+r'}+\delta_{i,j+r'+1}$ with indices 
read modulo $n.$ In other word, the quiver corresponding to the exchange matrix $B$ is a bipartite graph 
with $2n$ vertices labeled by $p_{1},\dots,p_{n},q_{1},\dots,q_{n}.$ There are four arrows adjacent to each 
$q_{i}$: two outgoing arrows from $q_{i}$ to $p_{i-r}$ and $p_{i+r'},$ two incoming arrows from $p_{i-r-1}$ 
and $p_{i+r'+1}$ to $q_{i}.$ This is called a \emph{generalized Glick's quiver} $\mathcal{Q}_{k,n}$, see 
Figure~\ref{fig: genglickquiver}. The higher pentagram map $T_\k$ is then a composition of a sequence of mutations
on all the $p_{i}-$vertices and a relabeling $\{p_{i}\mapsto q_{i+r'-r},q_{i}\mapsto p_{i}\}$\cite{GSTV}.

\begin{figure}
% Generated with LaTeXDraw 2.0.8
% Mon Dec 30 14:34:26 CST 2013
% \usepackage[usenames,dvipsnames]{pstricks}
% \usepackage{epsfig}
% \usepackage{pst-grad} % For gradients
% \usepackage{pst-plot} % For axes
\scalebox{1} % Change this value to rescale the drawing.
{
\begin{pspicture}(0,-1.3781251)(5.6828127,1.3781251)
\pscircle[linewidth=0.04,dimen=outer](2.567656,0.72874993){0.20906238}
\usefont{T1}{ptm}{m}{n}
\rput(2.6223435,1.1896875){$q_i$}
\psline[linewidth=0.04cm,arrowsize=0.05291667cm 2.0,arrowlength=1.4,arrowinset=0.0]{->}(1.0385938,-0.54031247)(2.4385939,0.5996876)
\usefont{T1}{ptm}{m}{n}
\rput(3.5423436,-1.1503125){$p_{i+r'}$}
\usefont{T1}{ptm}{m}{n}
\rput(4.6723437,-1.1503125){$p_{i+r'+1}$}
\usefont{T1}{ptm}{m}{n}
\rput(0.8123438,-1.1503125){$p_{i-r-1}$}
\usefont{T1}{ptm}{m}{n}
\rput(1.9623437,-1.1503125){$p_{i-r}$}
\psline[linewidth=0.04cm,arrowsize=0.05291667cm 2.0,arrowlength=1.4,arrowinset=0.0]{->}(2.4985938,0.55968755)(1.9785938,-0.50031245)
\psline[linewidth=0.04cm,arrowsize=0.05291667cm 2.0,arrowlength=1.4,arrowinset=0.0]{->}(2.6585937,0.55968755)(3.1985939,-0.50031245)
\pscircle[linewidth=0.04,dimen=outer](0.9076561,-0.69125){0.20906238}
\pscircle[linewidth=0.04,dimen=outer](1.8876561,-0.69125){0.20906238}
\pscircle[linewidth=0.04,dimen=outer](3.267656,-0.69125){0.20906238}
\pscircle[linewidth=0.04,dimen=outer](4.2276564,-0.69125){0.20906238}
\psline[linewidth=0.04cm,arrowsize=0.05291667cm 2.0,arrowlength=1.4,arrowinset=0.0]{->}(4.0785937,-0.54031247)(2.7385938,0.5996876)
\end{pspicture} 
}
\caption{The quiver $\mathcal{Q}_{\k,n}$ at $q_i$}
\label{fig: genglickquiver}
\end{figure}
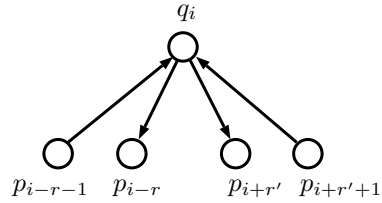

\subsection{Higher pentagram maps and $Y-$systems} 
We now use the wrapping of the generalized Glick's quiver $\mathcal{Q}_{\k,n}$ around a torus, interpreted in \cite{DFKT} 
as the ``octahedron" quiver with vertices in $\Z^2$
of Fig. \ref{octaquiver}, and with suitable identification of vertices, along the two periods of the torus:
$(i,j)\equiv (i+\k,j+2-\k)$ and $(i,j)\equiv (i+n,j-n)$. We show a sample of such a wrapping in Fig.~\ref{Fig: GlickTorus}
for $\k=3$ and $n=5$.

\begin{figure}
\input{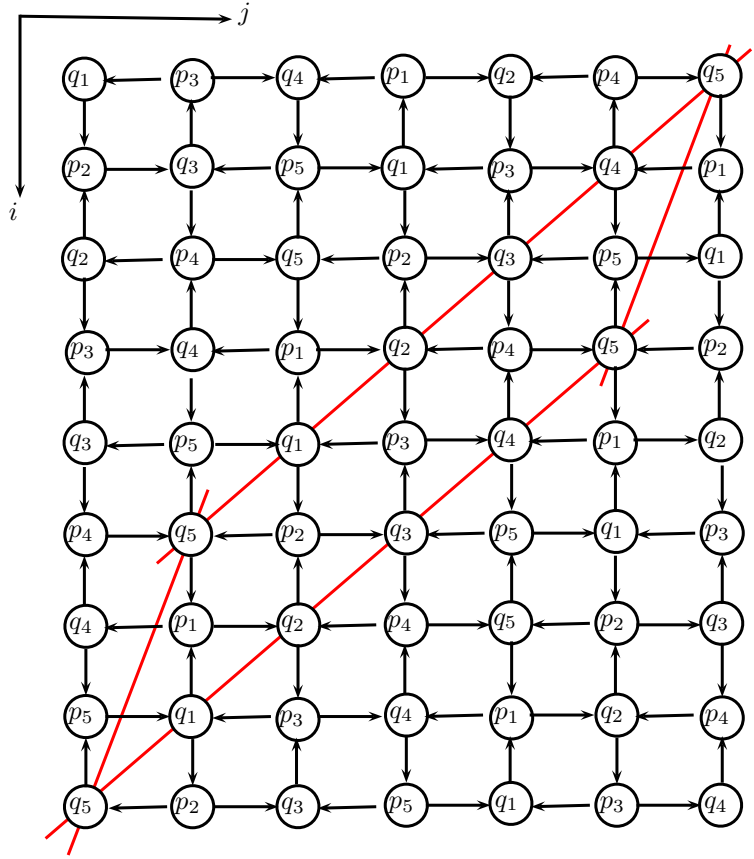}
\caption{Glick's quiver for the pentagram map, viewed as a torus-wrapped octahedron quiver for $\k=3$ and $n=5$. 
We have represented a fundamental domain by solid lines, along the two periods $(3,-1)$ and $(5,-5)$.}
\label{Fig: GlickTorus}
\end{figure}

%\begin{figure}
%\input{fig_genglicktorus}
%\caption{$\mathcal{Q}_{5,8}$ on torus}
%\label{fig: genglicktorus}
%\end{figure}

Having established this connection, upon carefully comparing the transformation \eqref{highpentaT} with the $Y$-system evolution
\eqref{Ysys} ,
it is now easy to interpret the $p,q$ coordinates of the higher pentagram map in terms of $Y$-system solutions.
We arrive at the following:

\begin{prop}\label{Prop: YsystemHigherPent}
Let $A$ be a twisted $n-$gon in $\mathbb{RP}^{\k-1}$ with $p,q$ coordinates $(p_i(A),q_i(A))_{i\in [1,n]}$. 
Let $\{Y_{i,j,k}:i,j,k\in\mathbb{Z},i+j+k\equiv0\mod2\}$ be the $Y-$system solution subject to the initial conditions:
\begin{eqnarray}
Y_{i,j,-1}  &=&\left(q_{\left(((\k-2) i+\k j+r-r'\right)/2}(A)\right)^{-1}\qquad (i,j\in \Z;i+j=1\, {\rm mod}\, 2)\nonumber \\
Y_{i,j,0}  &=& p_{((\k-2) i+\k j)/2}(A)\qquad \qquad \qquad (i,j\in \Z;i+j=0\, {\rm mod}\, 2), \label{initinv}
\end{eqnarray}
Then, the $p,q$ coordinates of the $k$-th iterate $T_\k^k$ of the higher pentagram map on the polygon $A$ read:
\[
p_{((\k-2) i+\k j+k(r-r'))/2}\Big(T_\k^{k}(A)\Big)=Y_{i,j,k},\  q_{((\k-2) i+\k j+(k+1)(r-r'))/2}\Big(T_\k^{k}(A)\Big)=Y_{i,j,k-1}^{-1}
\]
for any $i,j,k\in\mathbb{Z}$ with respectively $i+j+k\equiv0\mod2$ and $i+j+k-1\equiv0\mod2$. 
\end{prop}

Recalling that the coordinates $p_i(A)$ and $q_i(A)$ are periodic with period $n$, so are $p_i\Big(T_\k^k(A)\Big)$
and $q_i\Big(T_\k^k(A)\Big)$. We deduce the following:

\begin{cor}\label{douper}The $Y-$system solution in Proposition \ref{Prop: YsystemHigherPent} has double periodicity
$(\k,2-\k,0)$ and $(n,-n,0),$ i.e.,
\begin{equation}\label{perdou}
Y_{i,j,k}=Y_{(i,j,k)+\alpha (\k,2-\k,0)+\beta (n,-n,0)}
\end{equation}
for $i+j+k\equiv0\mod2$ and $\alpha,\beta\in\mathbb{Z}.$
\end{cor}

\subsection{From periodic Y-systems to quasi-periodic solutions of the octahedron relation}

We now unfold the $Y$-system solution of the previous section into a solution of the octahedron relation with suitable ``quasi-periodic" 
boundary conditions. The latter will be explicitly functions of the conserved quantities \eqref{consT}.
The main problem here is the ``inversion" of the relations
\begin{equation}\label{YtoT} Y_{i,j,k}=\frac{T_{i+1,j,k}T_{i-1,j,k}}{T_{i,j+1,k}T_{i,j-1,k}} 
\end{equation}
namely find $T$'s for given $Y$'s. Writing as before the simplest initial data for the octahedron relation in the form
$T_{i,j,i+j+1\, {\rm mod}\, 2}= x_{i,j}$, we want to find conditions on the variables $x_{i,j}$ which guarantee that the 
$Y$ variables \eqref{YtoT}, expressed in terms of the octahedron solution $T_{i,j,k}$, actually encode the $p,q$ invariants
of twisted polygons and their iterated images under the higher pentagram map. 

We saw in previous section that the corresponding $Y$ variables
must satisfy the double periodicity conditions of Corollary \ref{douper}. Conversely, given any solution of the $Y$-system with
such periodicity conditions, let us define $p,q$ invariants of a twisted polygon $A$ by \eqref{initinv}. We may
then interpret the $p$ and $q$ invariants of the $k$-th iterate of the higher pentagram map on $A$ in terms 
of the solution $Y_{i,j,k}$ and $Y_{i,j,k-1}^{-1}$.

We have the following:

\begin{thm}
Let $T_{i,j,k}$, $i+j+k=1$ mod 2, be the solution of the octahedron relation \eqref{Tsys} 
with initial conditions $T_{i,j,i+j+1\, {\rm mod}\, 2}= x_{i,j}$ for $i,j\in \Z$.
Assuming that 
\begin{eqnarray}
x_{i+\k,j+2-\k}&=&x_{i,j} \qquad (i,j\in \Z) \nonumber \\
x_{i+n,i-n}&=&x_{i,j}\times  \left\{ \begin{matrix}
\lambda^{(\k-2) i+\k j} & {\rm if}\, i+j=1\, {\rm mod}\, 2\\
\mu^{(\k-2) i+\k j}& {\rm if}\, i+j=0\, {\rm mod}\, 2\\
\end{matrix}\right. \label{initperiodi}
\end{eqnarray}
then the solution $T_{i,j,k}$ with $i+j+k=1$ mod 2 satisfies the same properties, namely $T_{i+\k,j+2-\k,k}=T_{i,j,k}$ and
\begin{equation}\label{Tperio}
T_{i+n,j-n,k}= T_{i,j,k}\times  \left\{ \begin{matrix}
\lambda^{(\k-2) i+\k j} & {\rm if}\, k=0\, {\rm mod}\, 2\\
\mu^{(\k-2) i+\k j}& {\rm if}\, k=1\, {\rm mod}\, 2\\
\end{matrix}\right.
\end{equation}
and the corresponding $Y_{i,j,k}$ of eq.\eqref{YtoT} is doubly periodic as in \eqref{perdou}. Moreover, the two integrals of motion
$O_n$ and $E_n$ of eq. \eqref{consT} are given by: $O_n=\lambda^{2\k-2}$ and $E_n=\mu^{2-2\k}$.
\end{thm}
\begin{proof}
We must choose a fundamental domain of the integer plane $\Z^2$ under translations by $(\k,2-\k)$ and $(n,-n)$. 
Let us pick two parallel lines $P_n=\{(i,-i),(i+1,-i)\}_{i\in [0,n-1]}$, and consider the periodicity condition in $Y$ that identifies $i\equiv i+n$
(see Figure \ref{Fig: GlickTorus} for an illustration: the two parallel lines are made of vertices 
$q_1,q_2,...,q_5$ and $p_1,p_2,...,p_5$ respectively).
Using the relation \eqref{YtoT}, we see that the initial data for the $Y$-system now satisfy:
$Y_{i+n,-i-n,k} =Y_{i,-i,k}$ and $Y_{i+\k,j+2-\k,k}=Y_{i,j,k}$ for $k=0,1$ by direct substitution of \eqref{initperiodi},
and moreover we can compute:
\begin{eqnarray}
\prod_{i=0}^{n-1} Y_{i,-i,0}&=&\prod_{i=0}^{n-1} \frac{x_{i+1,-i}x_{i-1,-i}} {x_{i,-i+1}x_{i,-i-1}}
= \frac{x_{-1,0}x_{n,-n+1}}{x_{0,1}x_{n-1,-n}}=\lambda^{2\k-2}\label{evencons} \\
\prod_{i=0}^{n-1} Y_{i+1,-i,1}^{-1}&=&\prod_{i=0}^{n-1} \frac{x_{i+1,-i+1}x_{i+1,-i-1}}{x_{i+2,-i}x_{i,-i}} 
= \frac{x_{1,1}x_{n,-n}}{x_{0,0}x_{n+1,-n+1}}=\mu^{2-2\k} \label{oddcons}
\end{eqnarray}
The equation \eqref{Tperio} follows immediately by induction, using the octahedron equation, and the double periodicity of the $Y_{i,j,k}$
follows by direct substitution. Finally, we recover that the products of even $Y$'s and that of odd $Y$'s are separately conserved modulo the octahedron equation, which confirms the two conserved quantities \eqref{consT}. Their values are given by (\ref{evencons}-\ref{oddcons}).
\end{proof}

\begin{remark}
As noted in \cite{GSTV}, the ordinary pentagram case studied by Glick has the two conserved quantities $O_n$, $E_n$ related via $O_nE_n=1$,
which would correspond here to choosing $\lambda,\mu$ such that $\lambda=\mu$. However, in general, the two conserved quantities can have arbitrary values.
\end{remark}

\subsection{Higher dimensional generalizations of the pentagram map}
The higher pentagram maps of \cite{GSTV} are to be distinguished from the
higher dimensional generalizations of the pentagram maps, which are maps on $n$-gons in $d$-dimensional projective space \cite{KS}. These maps are also integrable in the continuous limit, and in some cases are shown to be Adler-Gelfand-Dickii flows. Integrability is shown in these cases by presenting Lax representations with a spectral parameter \cite{KS,KS2,MB,MB2}. 

A twisted $n-$gon in $\mathbb{RP}^{d}$ is a sequence $(v_j)_{j\in\mathbb{Z}}$ in $\mathbb{RP}^{d}$ with a monodromy matrix $M\in PSL_{d+1}(\mathbb{R})$ such that $v_{j+n}=M \circ v_j$ for all $j\in \mathbb{Z}.$ We can lift it to a sequence $(V_j)_{j\in\mathbb{Z}}$ in $\mathbb{R}^{d}$ satisfying $\det(V_j V_{j-1} \dots V_{j-(d-1)})=1$ for all $j\in\mathbb{Z}.$ This implies a linear relation of the form
\begin{equation}\label{linearKS}
V_j = a_{j,1} V_{j+1} + a_{j,2} V_{j+2} + \cdots + a_{j,d} V_{j+d} + (-1)^d V_{j+d+1}, \ j\in \Z,
\end{equation}
The coefficients $a_{j,i}$ are periodic: $a_{j+n,i}=a_{j,i}$
for some $a_{j,l}\in\R$.  In the next section we will connect the coefficients of this equation with the $T$-system with special boundary conditions. The periodicity of the coefficients in the linear recursion relation will turn out to be a manifestation of the Zamolodchikov periodicity phenomonon for the $q$-characters of the Lie algebra $A_d$.

\section{The $T$-system with special boundary conditions}
\subsection{The $A_d$ $T$-system}
The octahedron relation \eqref{Tsys} is a relation between variables on $\Z^3$. There are several important examples of interesting boundary conditions for this equation. The first one comes from the representation theory of the quantum affine algebra $U_q(\widehat{sl}_{d+1})$. 

The $q$-characters \cite{FrenResh} of $U_q(\widehat{\sl}_{d+1})$ form a commutative algebra which generalizes that satisfied by the characters of $\sl_{d+1}$. 
These $q$-characters satisfy the $A_d$ $T$-system with special initial data \cite{Nakajima}. (This initial data does not play a role  in the context of this paper, so we sometimes refer to $q$-characters as the solutions of the $A_d$ $T$-system absent the specialized initial data.) We therefore refer to the $A_d$ $T$-system as the octahedron relation \eqref{Tsys} subject to the following boundary conditions:
\begin{equation}\label{ArTsys}
T_{0,j,k}=T_{d+1,j',k'}=1,\qquad j,j',k,k'\in \Z.
\end{equation} 
It is immediate from Equation \eqref{Tsys} that this implies that $T_{-1,j,k}=T_{d+2,j,k}=0$ for all $j,k$. The $T$-system is an therefore a discrete evolution equation which takes place in in a subset of $\Z^3$ consisting of a strip, defined by $0\leq i \leq d+1$ and $j,k\in \Z$.

\begin{remark}
The meaning of the indices $(i,j,k)$ in the representation theory of the quantum affine algebra is as follows. The index $i$ corresponds to one of the $r$ simple roots of $\sl_{d+1}$. The index $k$ is related to the shift in the spectral parameter carried by finite-dimensional $U_q(\widehat{\sl}_{d+1})$-modules. The $q$-character of $T_{i,j,k}$ is that of the Kirillov-Reshetikhin module which, in the classical limit, has highest weight $j \omega_i$ and spectral paramater $z q^{2k}$ for a fixed non-zero complex number $z$.
\end{remark}

\subsection{Pl\"ucker relations and conserved quantities}
The $A_d$ $T$-system is a special Pl\"ucker relation called the Desnanot-Jacobi relation or Dodgson condensation. This is a relation which gives the determinant of an $n\times n$ matrix in terms of determinants of $(n-1)\times (n-1)$ and $(n-2)\times (n-2)$ matrices. 

Let $M$ be an $n\times n$ matrix and let $M_{i_1,...,i_k}^{j_1,...,j_k}$ be the $(n-k)\times(n-k)$ minor obtained be {\em deleting} the rows $i_1,...,i_k$ and columns $j_1,...,j_k$.  Then the relation is that
\begin{equation}\label{Desnanot}
| M | | M_{1,n}^{1,n} | = |M_{1}^1| |M_n^n| - |M_1^n| | M_n^1|.
\end{equation}
This is a discrete recursion with the natural initial data that the $0\times 0$ determinant is 1 and that the $1\times 1$ determinant is the single entry in the $1\times 1$ matrix.

If $T_{0,j,k}=1$, then this relation is satisfied by solutions of the $A_d$ $T$-system.
To see this, write an arbitrary $(i+1)\times (i+1)$ matrix $M^{(i+1)}_{j,k}$ with the following notation:
\begin{equation}\label{M}
M^{(i+1)}_{j,k}= \begin{pmatrix}
x_{j,k-i} & x_{j+1,k-i+1} & \cdots& x_{j+i,k} \\
x_{j-1,k-i+1} & x_{j,k-i+2} & \cdots &x_{j+i-1,k+1} \\
\vdots & & \ddots & \vdots \\
x_{j-i,k} & x_{j-i+1,k+1} & \cdots & x_{j,k+i}
\end{pmatrix}
\end{equation}
where for economy of space we use the variables $x_{j,k}=T_{1,j,k}$. 

With this notation, we can make the identification of the minors of $M=M^{(i+1)}_{j,k}$: are (with $n=i+1$)
\begin{eqnarray*}
M_{1,n}^{1,n}&=&M^{(i-1)}_{j,k},\qquad M_1^1= M^{(i)}_{j,k+1},\qquad M_n^n=M^{(i)}_{j,k-1},\\
& &\qquad M_1^n=M^{(i)}_{j-1,k},\qquad M_n^1=M^{(i)}_{j+1,k}.
\end{eqnarray*}
With the boundary condition that $T_{0,j,k}=1$ we conclude that $T_{i,j,k}=M^{(i)}_{j,k}$ satisfies the $T$-system when $i\geq0$.
Therefore we conclude that $T_{i,j,k}$ is determinant of a matrix of size $i$, $M^{(i)}_{j,k}$. 

In the $A_d$ $T$-system, we impose the additional boundary condition that $T_{d+1,j,k}=1$, which, upon inspection, implies that $T_{d+2,j,k}=0$. This last condition means that a determinant of a matrix of size $d+2$ vanishes. When we expand this matrix along a row or a column, we get two linear recursion relations for the variables $T_{1,*,*}$ with $d+2$ terms along the two directions $j+k$ and $k-j$. We claim these the coefficients in these recursion relations are the discrete integrals of the motion in the two directions of the discrete evolution, $j+k$ and $j-k$. 

In one direction, consider the two matrices with determinant 1, $M^{(d+1)}_{j,k}$ and $M^{(d+1)}_{j+1,k+1}$. These two matrices have $d$ columns in common with each other: Only the first column of  $M^{(d+1)}_{j,k}$ does not appear among the columns of $M^{(d+1)}_{j+1,k+1}$, and only the last column of the latter does not appear among the columns of the first.

Let $V_{a}$ be the $d+1$-dimensional vector making up the first column of the matrix $M^{(d+1)}_{j,k}$. Here, $a=\lfloor \frac{j+k}{2}\rfloor-d$ is determined by the sum of the two indices of the variables in the first columns: This sum is a constant along each column. That is, $V_a = (x_{j,k-d},x_{j-1,k-d+1}, \cdots, x_{j-d,k})^t$.

With this notation, 
$$ M^{(d+1)}_{j,k} = (V_{a}, V_{a+1}, \cdots, V_{a+d}), \qquad M^{(d+1)}_{j+1,k+1} = (V_{a+1}, V_{a+2}, \cdots, V_{a+d+1}).$$
Therefore,
\begin{eqnarray*}
 |M^{(d+1)}_{j+1,k+1}| &=& 1=|(V_{a+1}, V_{a+2}, \cdots, V_{a+d+1})|\\
 &=&(-1)^d|(V_{a+d+1},V_{a+1}, V_{a+2}, \cdots, V_{a+d})|.
 \end{eqnarray*}
Taking the difference between the two determinants:
\begin{eqnarray*}
0&=& 1-1 = |M^{(d+1)}_{j,k}|-|M^{(d+1)}_{j+1,k+1}| \\
&=& |(V_{a}, V_{a+1}, \cdots, V_{a+d})| - (-1)^d |(V_{a+d+1},V_{a+1}, V_{a+2}, \cdots, V_{a+d})| \\ 
&= &
| (V_{a}- (-1)^dV_{a+d+1}, V_{a+1}, \cdots, V_{a+d})|,
\end{eqnarray*}
where in the last line we use the linearity property of the determinant in its columns.

We conclude that there is a linear relation between the columns of the  matrix in the last line.
\begin{thm}
The columns of an arbitrary $(d+1)\times (d+1)$ matrix of determinant 0, such that each of its solid minors has determinant 1, satisfy the linear recursion relation
\begin{equation}\label{Vrecursion}
0=V_a + \sum_{i=1}^{d} (-1)^i \alpha_{a,i} V_{a+i} - (-1)^dV_{a+d+1}.
\end{equation}
\end{thm}
We note the similarity of this relation to Equation \eqref{linearKS}. The difference so far is the fact that there is no periodicity of the coefficients $\alpha_{a,i}$. We will add this in the following sections, interpreting the various ingredients in terms of the $T$-system.

We conclude that an ordered sequence of points in $\P\R^{d}$, lifted to $\R^{d+1}$ under the the restriction that every $(d+1)\times (d+1)$ determinant of the vectors corresponding to neighboring points is equal to 1, is given by the columns of $M^{(d+1)}_{j,k}$ and satisfy a linear recursion relation of the form \eqref{Vrecursion}. Moreover the components of these vectors correspond to solutions $T_{1,j,k}$ for various $j,k$ of the $A_d$ $T$-system.

Before doing so, let us consider the question of discrete integrability of the $T$-system which follows from these recursion relations.
On the other hand we can consider this as a relation between the components of the vectors $V_b$, that is, the entries of the matrix $M$. The relation holds for each of the rows of the vectors.
The $(b+1)$-st component of the relation is
\begin{equation}\label{linearrecursion}
0=x_{j-b,k-d+b} + \sum_{i=1}^{d} (-1)^i \alpha_{a,i} x_{j-b+i,k+b+i-d} - (-1)^d x_{j-b+d+1,k+b+1} ,\quad 0\leq b \leq d,
\end{equation}
where {\em the coefficients $\alpha_{a,i}$ are independent of the row $b$}, that is, they are independent of $j-k$. 

To summarize, for any $j,k$ we have
\begin{lemma}
There is a linear recursion relation satisfied by entries of the matrix $M$:
\begin{equation}\label{recursion}
0 = x_{j,k-d} + \sum_{i=1}^d (-1)^i c_{i}(j+k) x_{j+i,k+i-d} - (-1)^d x_{j+d+1,k+1},
\end{equation}
where the coefficients $c_i(j+k)$ are independent of the difference $j-k$.
\end{lemma}

Similarly, if we compare the matrices $M^{(r+1)}_{j,k}$ and $M^{(r+1)}_{j-1,k+1}$ we will see they differ by one row, and we will find another relation between the $x_{j,k}$'s with coefficients which are independent of $j+k$.  That is,
\begin{lemma}
There is a linear recursion relation satisfied by entries of the matrix $M$:
$$
0 = x_{j,k-d} + \sum_{i=1}^d (-1)^i d_i(j-k) x_{j-i,k+i-d} - (-1)^d x_{j-d-1,k+1},
$$
where the coefficients $d_i(j-k)$ are independent of the sum $j+k$.
\end{lemma}

The coefficients $c$ and $d$ therefore have the interpretation of constants of the motion in the direction $j-k$ and $j+k$, respectively. The $A_d$ $T$-system is a discrete integrable system.

\subsection{Wall boundary conditions}
In this section, we consider only the $A_d$ $T$-system, that is, the octahedron relation with the boundary conditions \eqref{ArTsys} imposed on it.
We now consider the effects of further boundary conditions, in a perpendicular direction, on the solutions of the $A_d$ $T$-system. 
We call these ``wall" boundary condition. First, we impose the following conditions:
\begin{equation}\label{wall}
T_{i,0,k}=1, \qquad i,k\in\Z.
\end{equation}

\begin{remark}
There is now a question of compatibility of boundary conditions here: In fact, one can show \cite{DFKT} that simply setting initial data $T_{i,0,1}=T_{i,0,0}=1$ implies the relation \eqref{zeros}  for all  $k$. This statement as well as the Theorems quoted below are all proved using the network solution of the octahedron relation.
\end{remark}

From the octahedron relation \eqref{Tsys} it is immediate that \eqref{wall} implies that $T_{i,-1,k}=0$ for all $i,k.$ Moreover, although it is not immediately obvious from the $T$-system itself, it was shown in \cite{DFKT} that
\begin{thm}
The solutions of the equation \eqref{Tsys} with boundary conditions \eqref{ArTsys} and \eqref{wall} satisfy
\begin{equation}\label{zeros}
T_{i,-j,k}=0, \qquad 0\leq j \leq d,\  i,k\in\Z.
\end{equation}
\end{thm}
Furthermore, there is a ``mirroring" phenomenon \cite{DFKT}:
\begin{thm}
The variables on one side of the ``wall of 1's" are determined by the variables on the other side as follows:
\begin{equation}\label{mirror}
T_{i,j,k} = (-1)^{d i} T_{d+1-i,-j-d-1,k}.
\end{equation}
\end{thm}

An illustration of this for $\sl_3$ is shown in Figure \ref{wallofones}. In this case, we can see the $i=1$ and $i=2$ planes in $\Z^3$ can be viewed as projected to the same plane, as $j+k\in 2\Z+1$ for $i=1$ and $j+k\in 2\Z$ for $i=2$, so they cover complementary sublattices in $\Z^2$. Since $d=2$, the minus sign from Equation \eqref{mirror} does not contribute to the variables on the left of the picture.
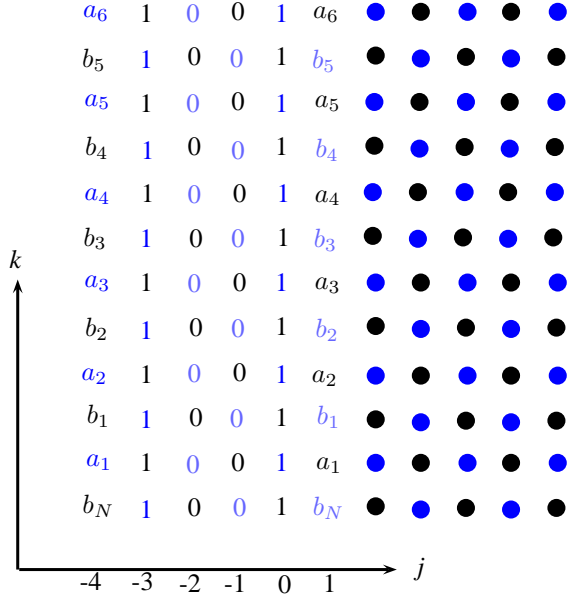
\begin{figure}
% Generated with LaTeXDraw 2.0.8
% Wed Nov 13 12:19:06 CST 2013
% \usepackage[usenames,dvipsnames]{pstricks}
% \usepackage{epsfig}
% \usepackage{pst-grad} % For gradients
% \usepackage{pst-plot} % For axes
\scalebox{1} % Change this value to rescale the drawing.
{
\begin{pspicture}(0,-3.99)(9.94,3.99)
\definecolor{color5302}{rgb}{0.4,0.4,1.0}
\definecolor{color5344}{rgb}{0.0,0.0,1.0}
\usefont{T1}{ptm}{m}{n}
\rput(2.47,1.995){0}
\usefont{T1}{ptm}{m}{n}
\rput(3.05,2.595){0}
\usefont{T1}{ptm}{m}{n}
\rput(2.49,0.775){0}
\usefont{T1}{ptm}{m}{n}
\rput(3.07,1.375){0}
\usefont{T1}{ptm}{m}{n}
\rput(2.49,-0.405){0}
\usefont{T1}{ptm}{m}{n}
\rput(3.07,0.195){0}
\usefont{T1}{ptm}{m}{n}
\rput(2.49,-1.605){0}
\usefont{T1}{ptm}{m}{n}
\rput(3.07,-1.005){0}
\usefont{T1}{ptm}{m}{n}
\rput(2.47,3.195){0}
\usefont{T1}{ptm}{m}{n}
\rput(3.05,3.795){0}
\usefont{T1}{ptm}{m}{n}
\rput(2.47,-2.805){0}
\usefont{T1}{ptm}{m}{n}
\rput(3.05,-2.205){0}
\usefont{T1}{ptm}{m}{n}
\rput(4.26,-2.245){$a_1$}
\usefont{T1}{ptm}{m}{n}
\rput(4.2,-1.085){$a_2$}
\usefont{T1}{ptm}{m}{n}
\rput(4.24,0.155){$a_3$}
\usefont{T1}{ptm}{m}{n}
\rput(4.24,1.315){$a_4$}
\usefont{T1}{ptm}{m}{n}
\rput(4.24,2.575){$a_5$}
\usefont{T1}{ptm}{m}{n}
\rput(4.22,3.755){$a_6$}
\usefont{T1}{ptm}{m}{n}
\rput(4.25,-1.605){\color{color5302}$b_1$}
\usefont{T1}{ptm}{m}{n}
\rput(4.23,-0.445){\color{color5302}$b_2$}
\usefont{T1}{ptm}{m}{n}
\rput(4.21,0.755){\color{color5302}$b_3$}
\usefont{T1}{ptm}{m}{n}
\rput(4.23,1.955){\color{color5302}$b_4$}
\usefont{T1}{ptm}{m}{n}
\rput(4.19,3.135){\color{color5302}$b_5$}
\usefont{T1}{ptm}{m}{n}
\rput(4.25,-2.805){\color{color5302}$b_N$}
\usefont{T1}{ptm}{m}{n}
\rput(2.45,2.575){\color{color5302}0}
\usefont{T1}{ptm}{m}{n}
\rput(3.03,3.175){\color{color5302}0}
\usefont{T1}{ptm}{m}{n}
\rput(2.47,1.355){\color{color5302}0}
\usefont{T1}{ptm}{m}{n}
\rput(3.05,1.955){\color{color5302}0}
\usefont{T1}{ptm}{m}{n}
\rput(2.47,0.175){\color{color5302}0}
\usefont{T1}{ptm}{m}{n}
\rput(3.05,0.775){\color{color5302}0}
\usefont{T1}{ptm}{m}{n}
\rput(2.47,-1.025){\color{color5302}0}
\usefont{T1}{ptm}{m}{n}
\rput(3.05,-0.425){\color{color5302}0}
\usefont{T1}{ptm}{m}{n}
\rput(2.45,3.775){\color{color5302}0}
\usefont{T1}{ptm}{m}{n}
\rput(3.07,-2.805){\color{color5302}0}
\usefont{T1}{ptm}{m}{n}
\rput(2.45,-2.225){\color{color5302}0}
\usefont{T1}{ptm}{m}{n}
\rput(3.03,-1.625){\color{color5302}0}
\usefont{T1}{ptm}{m}{n}
\rput(3.63,3.775){\color{blue}1}
\usefont{T1}{ptm}{m}{n}
\rput(3.63,2.575){\color{blue}1}
\usefont{T1}{ptm}{m}{n}
\rput(3.65,1.375){\color{blue}1}
\usefont{T1}{ptm}{m}{n}
\rput(3.63,-1.025){\color{blue}1}
\usefont{T1}{ptm}{m}{n}
\rput(3.63,0.195){\color{blue}1}
\usefont{T1}{ptm}{m}{n}
\rput(3.63,-2.205){\color{blue}1}
\usefont{T1}{ptm}{m}{n}
\rput(3.63,3.195){1}
\usefont{T1}{ptm}{m}{n}
\rput(3.63,1.995){1}
\usefont{T1}{ptm}{m}{n}
\rput(3.65,0.795){1}
\usefont{T1}{ptm}{m}{n}
\rput(3.63,-1.605){1}
\usefont{T1}{ptm}{m}{n}
\rput(3.63,-0.385){1}
\usefont{T1}{ptm}{m}{n}
\rput(3.63,-2.785){1}
\usefont{T1}{ptm}{m}{n}
\rput(1.83,-2.825){\color{blue}1}
\usefont{T1}{ptm}{m}{n}
\rput(1.83,3.155){\color{blue}1}
\usefont{T1}{ptm}{m}{n}
\rput(1.85,1.955){\color{blue}1}
\usefont{T1}{ptm}{m}{n}
\rput(1.83,-0.445){\color{blue}1}
\usefont{T1}{ptm}{m}{n}
\rput(1.83,0.775){\color{blue}1}
\usefont{T1}{ptm}{m}{n}
\rput(1.83,-1.625){\color{blue}1}
\usefont{T1}{ptm}{m}{n}
\rput(1.83,3.775){1}
\usefont{T1}{ptm}{m}{n}
\rput(1.83,2.575){1}
\usefont{T1}{ptm}{m}{n}
\rput(1.85,1.375){1}
\usefont{T1}{ptm}{m}{n}
\rput(1.83,-1.025){1}
\usefont{T1}{ptm}{m}{n}
\rput(1.83,0.195){1}
\usefont{T1}{ptm}{m}{n}
\rput(1.83,-2.205){1}
\usefont{T1}{ptm}{m}{n}
\rput(1.2,-2.205){\color{color5344}$a_1$}
\usefont{T1}{ptm}{m}{n}
\rput(1.14,-1.065){\color{blue}$a_2$}
\usefont{T1}{ptm}{m}{n}
\rput(1.18,0.175){\color{blue}$a_3$}
\usefont{T1}{ptm}{m}{n}
\rput(1.18,1.335){\color{color5344}$a_4$}
\usefont{T1}{ptm}{m}{n}
\rput(1.18,2.595){\color{color5344}$a_5$}
\usefont{T1}{ptm}{m}{n}
\rput(1.16,3.775){\color{blue}$a_6$}
\usefont{T1}{ptm}{m}{n}
\rput(1.19,-1.585){$b_1$}
\usefont{T1}{ptm}{m}{n}
\rput(1.17,-0.425){$b_2$}
\usefont{T1}{ptm}{m}{n}
\rput(1.15,0.775){$b_3$}
\usefont{T1}{ptm}{m}{n}
\rput(1.17,1.975){$b_4$}
\usefont{T1}{ptm}{m}{n}
\rput(1.13,3.155){$b_5$}
\usefont{T1}{ptm}{m}{n}
\rput(1.19,-2.785){$b_N$}
\psline[linewidth=0.04cm,arrowsize=0.05291667cm 2.0,arrowlength=1.4,arrowinset=0.4]{->}(0.14,-3.63)(5.16,-3.63)
\usefont{T1}{ptm}{m}{it}
\rput(5.48,-3.625){j}
\psline[linewidth=0.04cm,arrowsize=0.05291667cm 2.0,arrowlength=1.4,arrowinset=0.4]{->}(0.12,-3.63)(0.12,0.23)
\usefont{T1}{ptm}{m}{it}
\rput(0.09,0.495){k}
\usefont{T1}{ptm}{m}{n}
\rput(3.67,-3.825){0}
\usefont{T1}{ptm}{m}{n}
\rput(4.27,-3.805){1}
\usefont{T1}{ptm}{m}{n}
\rput(3.01,-3.805){-1}
\usefont{T1}{ptm}{m}{n}
\rput(2.41,-3.805){-2}
\usefont{T1}{ptm}{m}{n}
\rput(1.78,-3.785){-3}
\usefont{T1}{ptm}{m}{n}
\rput(1.09,-3.785){-4}
\psdots[dotsize=0.24,linecolor=blue](4.88,-2.21)
\psdots[dotsize=0.24,linecolor=blue](4.88,-1.05)
\psdots[dotsize=0.24,linecolor=blue](4.88,0.19)
\psdots[dotsize=0.24,linecolor=blue](4.84,1.39)
\psdots[dotsize=0.24,linecolor=blue](4.86,2.59)
\psdots[dotsize=0.24,linecolor=blue](4.88,3.79)
\psdots[dotsize=0.24,linecolor=blue](5.48,-2.83)
\psdots[dotsize=0.24,linecolor=blue](5.48,-1.67)
\psdots[dotsize=0.24,linecolor=blue](5.48,-0.43)
\psdots[dotsize=0.24,linecolor=blue](5.44,0.77)
\psdots[dotsize=0.24,linecolor=blue](5.46,1.97)
\psdots[dotsize=0.24,linecolor=blue](5.48,3.17)
\psdots[dotsize=0.24,linecolor=blue](6.1,-2.21)
\psdots[dotsize=0.24,linecolor=blue](6.1,-1.05)
\psdots[dotsize=0.24,linecolor=blue](6.1,0.19)
\psdots[dotsize=0.24,linecolor=blue](6.06,1.39)
\psdots[dotsize=0.24,linecolor=blue](6.08,2.59)
\psdots[dotsize=0.24,linecolor=blue](6.1,3.79)
\psdots[dotsize=0.24](4.88,-2.79)
\psdots[dotsize=0.24](4.88,-1.63)
\psdots[dotsize=0.24](4.88,-0.39)
\psdots[dotsize=0.24](4.84,0.81)
\psdots[dotsize=0.24](4.86,2.01)
\psdots[dotsize=0.24](4.88,3.21)
\psdots[dotsize=0.24](5.48,-2.21)
\psdots[dotsize=0.24](5.48,-1.05)
\psdots[dotsize=0.24](5.48,0.19)
\psdots[dotsize=0.24](5.44,1.39)
\psdots[dotsize=0.24](5.46,2.59)
\psdots[dotsize=0.24](5.48,3.79)
\psdots[dotsize=0.24](6.08,-2.81)
\psdots[dotsize=0.24](6.08,-1.65)
\psdots[dotsize=0.24](6.08,-0.41)
\psdots[dotsize=0.24](6.04,0.79)
\psdots[dotsize=0.24](6.06,1.99)
\psdots[dotsize=0.24](6.08,3.19)
\psdots[dotsize=0.24,linecolor=blue](6.68,-2.83)
\psdots[dotsize=0.24,linecolor=blue](6.68,-1.67)
\psdots[dotsize=0.24,linecolor=blue](6.68,-0.43)
\psdots[dotsize=0.24,linecolor=blue](6.64,0.77)
\psdots[dotsize=0.24,linecolor=blue](6.66,1.97)
\psdots[dotsize=0.24,linecolor=blue](6.68,3.17)
\psdots[dotsize=0.24,linecolor=blue](7.3,-2.21)
\psdots[dotsize=0.24,linecolor=blue](7.3,-1.05)
\psdots[dotsize=0.24,linecolor=blue](7.3,0.19)
\psdots[dotsize=0.24,linecolor=blue](7.26,1.39)
\psdots[dotsize=0.24,linecolor=blue](7.28,2.59)
\psdots[dotsize=0.24,linecolor=blue](7.3,3.79)
\psdots[dotsize=0.24](6.68,-2.21)
\psdots[dotsize=0.24](6.68,-1.05)
\psdots[dotsize=0.24](6.68,0.19)
\psdots[dotsize=0.24](6.64,1.39)
\psdots[dotsize=0.24](6.66,2.59)
\psdots[dotsize=0.24](6.68,3.79)
\psdots[dotsize=0.24](7.28,-2.81)
\psdots[dotsize=0.24](7.28,-1.65)
\psdots[dotsize=0.24](7.28,-0.41)
\psdots[dotsize=0.24](7.24,0.79)
\psdots[dotsize=0.24](7.26,1.99)
\psdots[dotsize=0.24](7.28,3.19)
\usefont{T1}{ptm}{m}{n}
%\rput(8.68,3.775){\color{blue}$T^{(1)}_{j,k}$}
\usefont{T1}{ptm}{m}{n}
%\rput(8.62,0.795){$T^{(2)}_{j,k}$}
\end{pspicture} 
}
\caption{An illustration of the mirroring phenomenon for $d=2$. The circles denote the variables $T_{i,j,k}$ where $i\in \{1,2\}$ and $j>0$.}
\label{wallofones}
\end{figure}

\begin{remark}
Note that this is consistent with the linear recursion relations of
Lemmas
\ref{recursion} when we interpret it as a recursion relation for $T_{1,j,k}$:
\begin{equation}\label{Trecursion}
0=T_{1,j,k} + \sum_{i=1}^d (-1)^i c_i(j+k)T_{1,j+i,k+i} - (-1)^d T_{1,j+d+1,k+d+1},\quad j+k\in 2\Z+1,
\end{equation}
which has $d+2$ terms, with the first coefficient  equal to 1 and the last coefficient equal to $(-1)^{d+1}$. When $j=-d-1$, for any  $k$ and $i=1$, all the entries in the sum in the middle vanish, and we have $(-1)^{d} - (-1)^d = 0$.
\end{remark}

\subsection{Identification of conserved quantities with solutions of the $T$-system with wall boundary conditions}
We will now show that, under the ``wall" boundary conditions \eqref{wall}, the conserved quantities $c_i(j+k)$ in Equation \eqref{recursion} are in fact solutions of the $A_d$ $T$-system, in particular, those with $j=1$.
\begin{thm}
The coefficients $c_i(j+k)$ are equal to the values of $T_{i',1,k'}$ along the boundary $j=1$, with $i'=d+1-i$ and $k'=k+d+i$.
\end{thm}
\begin{proof}
Consider the $(d+2)\times (d+2)$ matrix $M^{(d+2)}_{1,k+d+1}$ under the boundary conditions \eqref{wall}. Due to Theorem \ref{zeros} about the vanishing of the solutions when $-d-1<j<0$ this matrix has the following form:
\begin{equation}
M:=M^{(d+2)}_{1,k+d+1}=\begin{pmatrix}
x_{1,k} & x_{2,k+1} & x_{3, k+2} & x_{4,k+3} & \cdots & x_{d+2,k+d+1}\\
1 &         x_{1,k+2} & x_{2,k+3} & x_{3,k+4} &\cdots & x_{d+1,k+d+2}\\
0 & 1 &                    x_{1,k+4} & x_{2,k+5} & \cdots & x_{d,k+d+3} \\
0 & 0 & 1 & x_{1,k+6} & \cdots & x_{d-1,k+d+4}\\
\ldots & & & \ddots &  \cdots & \\
0 & 0 & \cdots & 0 & 1 & x_{1,k+2(d+1)}
\end{pmatrix}
\end{equation}
where we denoted $x_{j,k}=T_{1,j,k}$. The determinant of this matrix is $T_{d+2,1,k+d+1}=0$ so it vanishes. Moreover it has a one-dimensional kernel, by definition of the $A_d$ $T$-system (the determinants of any of the solid $(d+1)\times (d+1)$ minors are equal to 1). We solve the nullspace equation
$$
M \vec a = \vec0
$$
 for $\vec a = (a_0,a_1,...,a_{d+1})$. 
%(\alpha_{r+1}, \alpha_r, ..., \alpha_{0})$.
On the one hand, the entries of this vector are the coefficients (up to normaliation) of the linear recursion relation:
$$\sum_{i=0}^{d+1} a_i x_{i+1,k+i}=0$$
so we may identify $a_i = (-1)^i c_i(k+d+1)$.

On the other hand, expanding the determinant along the first row, we
have that $a_i=(-1)^i |M_1^{(d+i-1)}|$. The determinant of the minor is
particularly simple to compute. It is equal to the determinant of the
$(d+1-i)\times (d+1-i)$-matrix $M^{(d+1-i)}_{1,k+d+i+2}$, which is, by
definition, $T_{d+1-i,1,k+d+i+2}$.

\end{proof}
\subsection{Linear recursion relations under Zamolodchikov periodicity}
We have now identified the coefficients in the linear recursion relation \eqref{linearKS} as solutions of the $T$-system, as long as we do not require them to be periodic. To see how periodicity enters the picture in the context of $T$-systems, we refer to the Zamolodchikov periodicity phenomenon.

We must impose one final additional boundary condition on the $A_d$ $T$-system, of the form of a second ``wall of 1's" of the same form as \eqref{wall}, but positioned at $j=\ell+1$ where $\ell\geq 2$. That is, we look at the octahedron equation restricted to the $A_d$ slice with boundary conditions \eqref{ArTsys}, together with the wall boundary conditions at $j=0$ \eqref{wall} and
\begin{equation}\label{wall2}
T_{i,\ell+1,k}=1.
\end{equation}

There remains an evolution of the $T$-system under these boundary conditions in the $k$-direction, taking place in a tube with walls at $j=0,\ell+1$ and $i=0,d+1$. A valid initial data for this system are the values of the function $T$ in the finite set of points ($d \times \ell$ of these) $\{(i,j,i+j\mod 2: 1\leq i \leq d, 1\leq j \leq \ell\}$. All other values of $T$ are positive Laurent polynomials in this initial data, due to the Laurent property of cluster algebras.

An important result about these boundary conditions, originally conjectured by Zamolodchikov \cite{Z} and later proven in various contexts \cite{DFKT,IIKKN}, is that there is a periodicity phenomenon:
\begin{thm}
The solutions of the $A_d$ $T$-system $T_{i,j,k}$ under the boundary conditions \eqref{wall} and \eqref{wall2} satisfy a periodicity phenomonon:
\begin{equation}
T_{i,j,k+p} = T_{d+1-i, \ell+1-j,k},  \qquad p=\ell+d+2
\end{equation}
so that
\begin{equation}\label{periodicity}
T_{i,j,k+2p} = T_{i,j,k}.
\end{equation}
\end{thm}
This can be visualized as a mirroring along the line $j=-(d+1)/2$ (as implied by Equation \eqref{mirror}) as well as $j=\ell+1+(d+1)/2$, by symmetry. See \cite{DFKT} for a detailed analysis using the network formulation.
\begin{figure}
\input{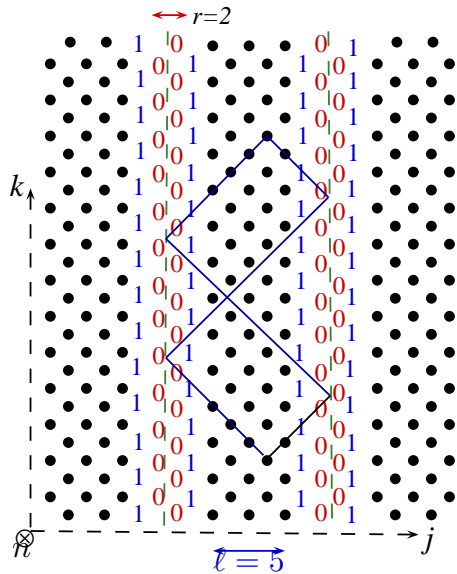}
\caption{An illustration of the how the mirroring along two lines implies periodicity for $d=2$, $\ell=5$. The $T$-variable at the bottom has two lines emanating from it, in the $j-k$ and $j+k$ directions, and these are mirrored a the dashed ``mirror" lines. Where they intersect at the top, the $T$-variable is equal to the original one at the bottom of the picture.}
\label{fig:periodic}
\end{figure}

In this situation, for fixed $\ell$ and $d$, we have a periodicity of the coefficients $c_i(j+k)$. We have identified the coefficients $\alpha_{a,i}$ in Equation \eqref{Vrecursion} with $T_{i',1,k}$, with $a=\lfloor \frac{j+k}{2}\rfloor -d$. Therefore we have
$$
\alpha_{a+p,i}=\alpha_{a,i}, \qquad p=\ell+d+2.
$$

Thus, the integrals of the motion in the case with wall boundary conditions are also periodic.
The number of integrals $\alpha_{a,i}$ is $p\times d= (\ell+d+2)d$. However these integrals are not algebraically independent in the periodic case, because they are determined by the set of initial data, which contains $\ell\times d$ indepenent data. Therefore, there are $d(d+2)$ relations between the variables $\alpha_{a,i}$. 

The vectors $V_a$ are also periodic in this case: $V_{a+p}=(-1)^d V_a$. They can be visualized as the collection of variables along a diagonal connecting $d+1$ points from NW to SE on the lattice of $T_{1,j,k}$'s. Then the set of such collections, translated with respect to each other by the vector $(1,1)$ in the $(j,k)$ direction, is periodic with period $p$ (recall that $a\sim (j+k)/2$).
We can therefore interpret the vectors $V_a$ as the lifts of the vertices of a closed $n$-gon in projective space.

\begin{remark}
There is a more general quasi-periodicity phenomonon corresponding to $n$-gons with monodromy for $T$-systems. This can be seen indirectly from the fact that (a) it is always possible to unfold the $Y$-variables corresponding to an arbitrary $T$-system such that periodic $Y$-systems unfold to quasi-periodic $T$-systems (as in Section 2) and (b) Zamolodchikov periodicity for $Y$-systems has been proven \cite{Keller}.
\end{remark}
\section{Conclusion}
In this paper, our interest was to exhibit two seemingly unrelated
relations between the $T$-system and the various versions of the
pentagram maps, which are all discrete integrable systems. The first
relation, discussed in Section 2, shows that the (generalized)
pentagram map in projective 2-space is in fact the octahedron relation
with special, quasi-periodic boundary conditions. Therefore
integrability follows from the fact that the octahedron relation is
known to be integrable. Integrability was proved in general for this
map by \cite{GSTV} by more classical means.

The second relation is simply the identification of the invariants
which are the periodic coefficients of the lifted coordinates of the
$n$-gon in projective $d$-space for general $d$ are solutions of the
$A_d$ $T$-system, that the fact that they are not algebraically
independent follows from the evolution of this $T$-system, and that
their periodicity is directly related to the Zamolodchikov periodicity
phenomenon.

The pentagram map is more generally defined for polygons with a
monodromy, not necessarily closed polygons corresponding to
periodicity. It would be interesting to see if there are compatible
boundary conditions for the $T$-system which reflect this more general
type of periodicity.

Moreover, the pentagrm map acts on projective variables, which are
related to ratios of solutions of the $T$-system. These satisfy the
$Y$-system. To get the most general type of $Y$-system solutions, one
must introduce coefficients into the $T$-system in a sufficiently
generic fashion. The solution of the $T$-system must be generalized in
this case. This will be the subject of a forthcoming publication.

%\bibliography{refs}
%\bibliographystyle{plain}

\end{document}